\documentclass[onecolumn,a4paper,11pt,accepted=2019-12-03]{quantumarticle}
\pdfoutput=1
\usepackage[utf8]{inputenc}
\usepackage[english]{babel}
\usepackage[T1]{fontenc}
\usepackage{hyperref}

\usepackage{ifdraft}
\ifdraft{\newcommand{\authnote}[3]{{\color{#3} {\bf  #1:} #2}}}{\newcommand{\authnote}[3]{}}
\newcommand{\anote}[1]{\authnote{ Andras}{#1}{red}}

\input{tikz_template.txt}

\usepackage{amsmath,amssymb,amsthm,mathtools}

\usepackage{mleftright}\mleftright
\usepackage{float}

\usepackage{enumitem}

\usepackage{verbatim}

\usepackage{tikz}
\usetikzlibrary{calc,arrows,quotes,angles}

\usepackage{thmtools} 
\usepackage{thm-restate}

\usepackage{multirow}
\usepackage{diagbox}

\floatstyle{boxed}

\newfloat{algorithm}{t}{lop}
\floatname{algorithm}{Algorithm} 

\newfloat{metaalgorithm}{t}{lop}
\floatname{metaalgorithm}{Meta-Algorithm} 

\newtheorem{theorem}{Theorem} 

\newtheorem{lemma}[theorem]{Lemma}

\newtheorem{corollary}[theorem]{Corollary}
\newtheorem{definition}[theorem]{Definition}

\mathchardef\mhyphen="2D

\newcommand{\ket}[1]{|#1\rangle}

\newcommand{\ipc}[2]{\left\langle #1 , #2 \right\rangle}

\newcommand{\opt}{\mbox{\rm OPT}}

\newcommand{\eps}{\varepsilon}
\renewcommand{\epsilon}{\varepsilon}
\newcommand{\nrm}[1]{\left\lVert#1\right\rVert}
\def\01{\{0,1\}}

\newcommand{\bigO}[1]{\mathcal{O}\left( #1 \right)}
\newcommand{\bOt}[1]{\widetilde{\mathcal O}\left(#1\right)}

\newcommand{\domf}{C}

\newcommand{\R}{\mathbb{R}}
\newcommand{\N}{\mathbb{N}}

\newcommand{\E}{\mathbb{E}}

\begin{document}
	
	\title{Convex optimization using quantum oracles}
	\author{Joran van Apeldoorn}
	\affiliation{QuSoft, CWI, Amsterdam, the Netherlands. {\tt apeldoor@cwi.nl}}
    \author{Andr\'as Gily\'en}
    \affiliation{QuSoft, CWI, Amsterdam, the Netherlands.
			{\tt gilyen@cwi.nl}}
    \author{Sander Gribling}
    \affiliation{QuSoft, CWI, Amsterdam, the Netherlands. {\tt gribling@cwi.nl}}
    \author{Ronald de Wolf}
    \affiliation{QuSoft, CWI and University of Amsterdam, the Netherlands.  {\tt rdewolf@cwi.nl}}
	
	\date{}
	\maketitle

	\begin{abstract}
		We study to what extent quantum algorithms can speed up solving convex optimization problems. Following the classical literature we assume access to a convex set via various oracles, and we examine the efficiency of reductions between the different oracles. In particular, we show how a separation oracle can be implemented using $\bOt 1$ quantum queries to a membership oracle, which is an exponential quantum speed-up over the $\Omega(n)$ membership queries that are needed classically.          We show that a quantum computer can very efficiently compute an approximate subgradient of a convex Lipschitz function. Combining this with a simplification of recent classical work of Lee, Sidford, and Vempala gives our efficient separation oracle. This in turn implies, via a known algorithm, that $\bOt n$ quantum queries to a membership oracle suffice to implement an optimization oracle (the best known classical upper bound on the number of membership queries is quadratic).        
        We also prove several lower bounds: $\Omega(\sqrt{n})$ quantum separation (or membership) queries are needed for optimization if the algorithm knows an interior point of the convex set, and $\Omega(n)$ quantum separation queries are needed if it does not.
\end{abstract}



\setcounter{page}{1}

\section{Introduction}

Optimization is a fundamental problem in mathematics and computer science, with many real-world applications. As people try to solve larger and larger optimization problems, the efficiency of optimization becomes more and more important, motivating us to find the best possible algorithms. Recent experimental progress on building quantum computers draws attention to new approaches to the problem: can we solve optimization problems more efficiently by exploiting quantum effects such as superposition, interference, and entanglement? For many discrete optimization problems~\cite{grover1996QSearch,durr1996QMinimumFinding,szegedy2004QMarkovChainSearch,durr2004QQueryCompGraph,ambainis2005QAlgMatchingNetwork} significant speed-ups have been shown, but less is known about \emph{continuous} optimization problems. 

One of the most successful continuous optimization paradigms is \emph{convex} optimization, which optimizes a convex function over a convex set that is given explicitly (by a set of constraints) or implicitly (by an oracle).
See Bubeck~\cite{bubeck2015ConvexOpt} for a recent survey. Quantum algorithms for convex optimization have been considered before. In 2008, Jordan~\cite{jordan2008PhDThesis} described a faster quantum algorithm for minimizing quadratic functions. 
Recently, for an important class of convex optimization problems (semidefinite optimization) quantum speed-ups were achieved using algorithms whose runtime scales polynomially with the desired precision and some geometric parameters~\cite{brandao2016QSDPSpeedup,apeldoorn2017QSDPSolvers,brandao2017QSDPSpeedupsLearning,apeldoorn2018ImprovedQSDPSolving}. However, many convex optimization problems can be solved classically using algorithms whose runtime scales \emph{logarithmically} with the desired precision and the relevant geometric parameters. We are aware of only one quantum speed-up which is partially in this regime, namely the very recent quantum interior point method of Kerenidis and Prakash~\cite{kerenidis2018QIntPoint}.
In this paper we look at general convex optimization problems, considering algorithms that have such favorable logarithmic scaling with the precision. 

The generic problem in convex optimization is minimizing a convex function $f:K \rightarrow \R \cup \{\infty\}$, where $K \subseteq \R^n$ is a convex set. We consider the setting where an interior point $x_0 \in \mathrm{int}(K)$ is given and radii $r,R >0$ are known such that $B(x_0,r) \subseteq K \subseteq B(x_0,R)$, where $B(x_0,r)$ is the Euclidean ball of radius $r$ centered at $x_0$. 

It is well-known that if the convex function is bounded on $K$, then we can equivalently consider the problem of minimizing a \emph{linear} function over a different convex set~$K'\subseteq\R^{n+1}$, namely the epigraph $K'=\{(x,\mu):x\in K, f(x)\geq\mu\}$ of~$f$. Accessing $K'$ is easy given access to $K$ and~$f$, and the parameters involved will be similar.
Conversely, for any linear optimization problem over an unknown convex set~$K$, there is an equivalent optimization problem over a known convex set (say, the ball), with an unknown bounded convex objective function~$f$ that can be evaluated easily given access to~$K$. 
From now on we therefore focus on optimizing a known linear function over an unknown convex set.

We consider the setting where access to the convex set is given only in a black-box manner, through an oracle. The five basic problems (oracles) in convex optimization identified by Gr{\"o}tschel, Lov\'asz, and Schrijver~\cite{grotschel1988GeomAlgAndConvOpt} are: membership, separation, optimization, violation, and validity (see Section~\ref{sec:prelim} for the definitions). They showed that all 
five basic problems are polynomial-time equivalent. That is, given an oracle~$O$ for one of these problems, one can implement an oracle for any of the other problems using a polynomial number of calls to~$O$ and polynomially many other elementary operations. Subsequent work made these polynomial-time reductions more efficient, reducing the degree of the polynomials. Recently Lee et al.~\cite{lee2017ConvexOptWMemb}, in the classical setting, showed that with \smash{$\bOt{n^2}$} calls\footnote{Here, and in the rest of the paper, the notation $\bOt{\cdot}$ is used to hide polylogarithmic factors in $n,r,R,\epsilon$.} to a membership oracle (and $\bOt{n^3}$ other elementary arithmetic operations) one can solve an optimization problem. They did so by showing that \smash{$\bOt{n}$} calls to a membership oracle suffice to do separation, 
and then composing this with the known fact~\cite
{lee2015FasterCuttingPlaneConvexOpt} (see also~\cite[Theorem~15]{lee2017ConvexOptWMemb}) that $\bOt{n}$ calls to a separation oracle suffice for optimization. 

Our main result (Section~\ref{sec:sep}) shows that on a quantum computer, $\bOt{1}$ calls to a membership oracle suffice to implement a separation oracle, and hence (by the known classical reduction from optimization to separation) $\bOt{n}$ calls to a membership oracle suffice for optimization.\footnote{Although not stated explicitly in our results, we also use $\bOt{n^3}$ additional operations for optimization using membership, like~\cite{lee2017ConvexOptWMemb}. This is because our quantum algorithm for separation uses only $\bOt{n}$ gates in addition to the $\bOt{1}$ membership queries, and we use the same reduction from optimization to separation as~\cite{lee2017ConvexOptWMemb}.
If queries themselves have significant time complexity, then our algorithm does lead to a speedup in time complexity over the best known classical algorithm. For example, if each membership query (with the required precision) takes time $\bOt{n^2}$ to implement, then our quantum algorithm for optimization has time complexity $\bOt{n^3}$, while the classical algorithm will use time $\bOt{n^4}$
because it uses $\bOt{n^2}$ membership queries.} Lee et al.~\cite{lee2017ConvexOptWMemb} use a geometric idea to reduce separation to finding an approximate subgradient of a convex Lipschitz function. They then show that $\bOt{n}$ evaluations of a convex Lipschitz function suffice to get an approximate subgradient. Our contributions here are twofold (Section~\ref{sec:approxsubgradient} and~\ref{sec:sep}). 
We use the same geometric idea, but we provide a simpler way to compute an approximate subgradient of a convex Lipschitz function (Section~\ref{sec:approxsubgradient}). We point out that this new algorithm is purely classical. Besides being simpler, the main advantage of our algorithm is that it is suitable  for a quantum speed-up using known quantum algorithms (Jordan's algorithm) for computing approximate (sub)gradients~\cite{jordan2005QuantGrad,gilyen2017OptQOptAlgGrad}, which we show in Section~\ref{sec:sep}. To show our quantum speed-up, we have to extend Jordan's quantum algorithm for gradient-computation to the case of convex Lipschitz functions. 


As a second set of results, in Section~\ref{sec:lowerbounds} we provide lower bounds on the number of membership or separation queries needed to implement several other oracles. We show that our quantum reduction from separation to membership indeed improves over the best possible classical reduction: $\Omega(n)$ classical membership queries are needed to do separation.\footnote{We are not aware of an existing proof of this classical lower bound, but it may well be somewhere in the vast literature on convex optimization.} 
We only have partial results regarding the optimality of the reduction from optimization to separation. In the setting where we are not given an interior point of the set~$K$, we can prove an essentially optimal $\Omega(n)$ lower bound on the number of quantum queries to a separation oracle needed to do optimization, using the general adversary bound. This lower bound implies that a quantum computer offers no query speed-up over a classical computer for the task of finding an interior point. 

However, for the case of quantum algorithms that \emph{do} know an interior point, we are only able to prove an $\Omega(\sqrt{n})$ lower bound. In the classical setting, regardless of whether or not we know an interior point, the reduction uses $\widetilde \Theta(n)$ queries. This raises the interesting question of whether knowing an interior point can lead to a better quantum algorithm.
We therefore view closing the gap between upper and lower bound as an important direction for future work. 

Finally, we briefly mention (Section~\ref{sec:corollaries}) how to obtain upper and lower bounds for some of the other oracle reductions, using a convex polarity argument. As we show, in the setting where we are given an interior point, the relation between membership and separation is analogous to the relation between validity and optimization. In particular, our better quantum algorithm for separation using membership queries implies that on a quantum computer $\bOt{1}$ queries to a validity oracle suffice to implement an optimization oracle. That is, on a quantum computer, finding the optimal value is equivalent to finding an optimizer. Also, the same polarity argument shows that algorithms for optimization using separation are essentially equivalent to algorithms for separation using optimization. In particular, this turns our lower bound on the number of separation queries needed to implement an optimization oracle into a lower bound on the reverse direction.

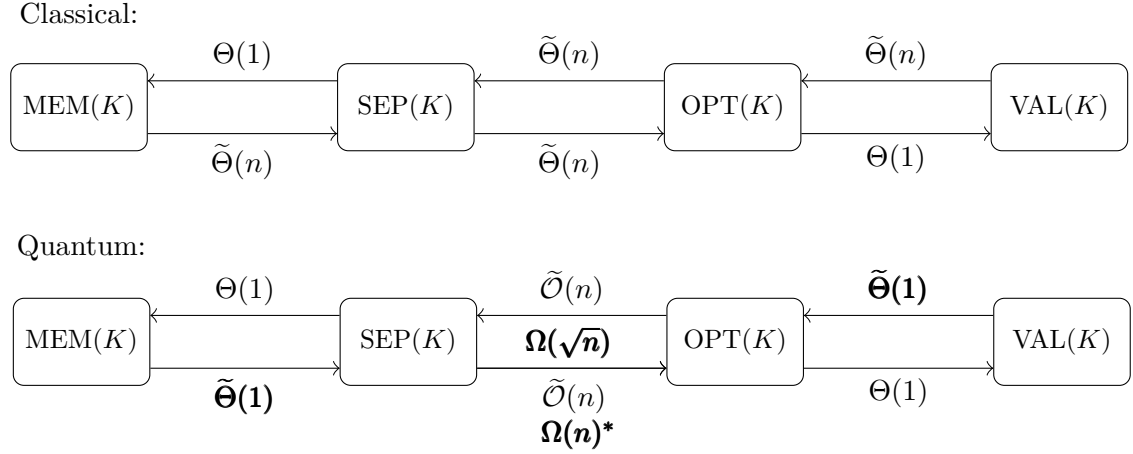
\begin{figure}[ht]
	\centering
	\begin{minipage}[t]{\textwidth}
		\centering
		\begin{tikzpicture}[node distance = 4.3cm]
		\tikzstyle{block} = [rectangle, draw, text width=4em, text centered, rounded corners, minimum height=3em]
		\tikzstyle{line} = [draw, -tonew,arrowhead=0.085cm]
		\tikzstyle{dot_line} = [draw, -tonew,dashed,arrowhead=0.085cm]
		
		\node [block] (mem) {\small $\mathrm{MEM}(K)$};
		\node [block, right of=mem] (sep) {\small $\mathrm{SEP}(K)$};
		\node [block, right of= sep] (opt) {\small $\mathrm{OPT}(K)$};
		\node [block, right of=opt] (val) {\small $\mathrm{VAL}(K)$};
		
		\node[above of=mem, node distance = 1.25cm] (classical) {Classical:};
		\path [line,transform canvas={yshift=-0.35cm}] (mem) to node[below] {$\widetilde \Theta(n)$} (sep);
		\path [line,transform canvas={yshift=0.35cm}] (sep) to node[above] {$\Theta(1)$} (mem);
		\path [line,transform canvas={yshift=-0.35cm}] (sep) to node[below] {$\widetilde \Theta(n)$}  (opt);
		\path [line,transform canvas={yshift=0.35cm}] (opt) to node[above] {$\widetilde \Theta(n)$}  (sep);
		\path [line,transform canvas={yshift=-0.35cm}] (opt) to node[below] {$\Theta(1)$}  (val);
		\path [line,transform canvas={yshift=0.35cm}] (val) to node[above] {$\widetilde \Theta(n)$}  (opt);
		\end{tikzpicture}
	\end{minipage}
	
	\vspace{0.8cm}
	
	\begin{minipage}[t]{\textwidth}
		\centering
		\begin{tikzpicture}[node distance = 4.3cm]
		\tikzstyle{block} = [rectangle, draw, text width=4em, text centered, rounded corners, minimum height=3em]
		\tikzstyle{line} = [draw, -tonew,arrowhead=0.085cm]
		\tikzstyle{dot_line} = [draw, -tonew,dashed,arrowhead=0.085cm]
		
		\node [block] (mem) {\small $\mathrm{MEM}(K)$};
		\node [block, right of=mem] (sep) {\small $\mathrm{SEP}(K)$};
		\node [block, right of= sep] (opt) {\small $\mathrm{OPT}(K)$};
		\node [block, right of=opt] (val) {\small $\mathrm{VAL}(K)$};
		
		\node[above of=mem, node distance = 1.25cm] (quantum) {Quantum:};
		\path [line,transform canvas={yshift=-0.35cm}] (mem) to node[below] {$\pmb{\widetilde \Theta(1)}$} (sep);
		\path [line,transform canvas={yshift=0.35cm}] (sep) to node[above] {$\Theta(1)$} (mem);
		\path [line,transform canvas={yshift=-0.35cm}] (sep) to node[below,text width = 1cm] {$\phantom{,}\bOt{n}$  \\ $\phantom{,}\pmb{\Omega(n)^*}\kern-10mm$}  (opt);
		\path [line,transform canvas={yshift=-0.35cm}] (sep) to node[above,text width = 1.2cm] {$\pmb{\Omega(\sqrt{n})}$}  (opt);   
		\path [line,transform canvas={yshift=0.35cm}] (opt) to node[above] {$\bOt{n}$}  (sep);
		\path [line,transform canvas={yshift=-0.35cm}] (opt) to node[below] {$\Theta(1)$}  (val);
		\path [line,transform canvas={yshift=0.35cm}] (val) to node[above] {$\pmb{\widetilde \Theta(1)}$}  (opt);
		\end{tikzpicture}%
	\end{minipage}
	\caption{The top and bottom diagram illustrate the relations between the basic (weak) oracles for respectively classical and quantum queries, with boldface entries marking our new results. All upper and lower bounds hold in the setting where we know an interior point of $K$, except the $*$-marked $\Omega(n)$ lower bound on the number of separation queries needed for optimization. Notice the central symmetry of the diagrams, which is a consequence of polarity.}
	\label{fig:results}
\end{figure}

Figure~\ref{fig:results} gives an informal presentation of our results; the upper bounds arise from oracle reductions, the (change in) accuracy is ignored here for simplicity. 
The above-mentioned polarity manifests itself in the central symmetry of the figure. 

\paragraph{Related independent work.}
In independent simultaneous work, Chakrabarti, Childs, Li, and Wu~\cite{chakrabarti2018QuantumConvexOpt} discovered a similar upper bound as ours: combining the recent classical work of Lee et al.~\cite{lee2017ConvexOptWMemb} with a quantum algorithm for computing gradients, they show how to implement an optimization oracle via $\bOt{n}$ quantum queries to a membership oracle and to an oracle for the objective function. Their proof stays quite close to~\cite{lee2017ConvexOptWMemb} while ours first simplifies some of the technical lemmas of~\cite{lee2017ConvexOptWMemb}, giving us a slightly simpler presentation and a better error-dependence of the resulting algorithm.
They also prove several lower bounds that are similar to the ones we prove here.

\section{Preliminaries} \label{sec:prelim}
We use $[n] := \{1, 2, \ldots, n\}$. 
For $p \geq 1$, $\epsilon \geq 0$, and a set $\domf \subseteq \R^n$ we let 
\[
B_p(\domf,\epsilon) = \{ x \in \R^n : \exists y \in \domf \text{ such that } ||x-y||_p \leq \epsilon \}
\]
be the set of points of distance at most $\epsilon$ from $\domf$ in the $\ell_p$-norm. When $\domf = \{x\}$ is a singleton set we abuse notation and write $B_p(x,\epsilon)$. We overload notation by setting 
\[
B_p(\domf,-\epsilon) = \{ x \in \R^n : B_p(x,\epsilon) \subseteq \domf \}.
\]
Whenever $p$ is omitted it is assumed that $p=2$. 

Recall that a function $f:\domf\rightarrow \R$ is \emph{Lipschitz} if there exists a constant $L>0$ such that 
\[
\left|f(y')-f(y)\right|\leq L\nrm{y'-y}_2 \text{ for all } y,y'\in \domf. 
\]
We write that $f$ is \emph{$L$-Lipschitz}.
The inner product between vectors $v,w\in\R^n$ is $\ipc{v}{w}=v^T w$.

\begin{definition}[Subgradient]
	Let $\domf \subseteq\R^n$ be convex and let $x$ be an element of the interior of $\domf$. For a convex function $f:\domf\rightarrow \R$ we denote by $\underline{\partial} f(x)$ the set of subgradients of $f$ at $x$, i.e., those vectors $g$ satisfying
	$$ f(y)\geq f(x) + \ipc{g}{y-x}\text{ for all } y\in \domf.
	$$
\end{definition}
\noindent Note that in the above definition $\underline{\partial} f(x)\neq \emptyset$ due to convexity.

If $f:\domf \rightarrow \R$ is $L$-Lipschitz, then for any $x$ in the interior of $\domf$ and any $g \in \underline \partial f(x)$ we have $\nrm{g} \leq L$, as follows. Consider a $y \in \domf$ such that $y-x = \alpha g$ for some $\alpha >0$. Then since $g$ is a subgradient of $f$ at $x$ we have 
\begin{equation}\label{eq:subgradBound}
\alpha \nrm{g}^2 = \ipc g{y-x} \leq f(y) - f(x) \leq L\nrm{y-x} = \alpha L \nrm{g},
\end{equation}
and therefore $\nrm{g} \leq L$. 

We will assume familiarity with quantum computing~\cite{nielsen2002QCQI}. In particular, a standard quantum oracle corresponds to a unitary transformation that acts on two (finite-dimensional) registers, where the first register contains the query and the answer is added to the second register. For example, a function evaluation oracle for $f\colon X \rightarrow Y$ would map $\ket{x,0}$ to $\ket{x,f(x)}$, where $\ket{x}$ and $\ket{f(x)}$ are basis states corresponding to binary representations of $x$ and $f(x)$ respectively.
Unlike classical algorithms, quantum computers can apply such an oracle to a \emph{superposition} of different $y$'s. They are also allowed to apply the inverse of a unitary oracle.

The standard quantum oracle described above models problems where there is a single correct answer to a query. When there are multiple good answers (for instance, different good approximations to the correct value) and the oracle is only required to give a correct answer with high probability, then we will work with the more liberal notion of \emph{relational} quantum oracles.

\begin{definition}[Relational quantum oracle] \label{def:reloracle}
	Let $\mathcal{F}\colon X \rightarrow \mathcal{P}(Y)$ be a function, such that for each $x\in X$ the subset $\mathcal{F}(x)\subseteq Y$ is the set of valid answers to an $x$ query. A relational quantum oracle for $\mathcal{F}$ which answers queries with success probability $\geq 1-\rho$, is a unitary that for all $x\in X$ maps
	$$
	U \colon \ket{x,0,0} \mapsto  \sum_{y\in Y}\alpha_{x,y}\ket{x,y,\psi_{x,y}},
	$$
	where $\ket{\psi_{x,y}}$ denotes some normalized quantum state 
	and $\sum_{y\in \mathcal{F}(x)}|\alpha_{x,y}|^2\geq 1- \rho$.
	Thus measuring the second register of $U\ket{x,0,0}$ gives a valid answer to the $x$ query with probability at least $1-\rho$.
\end{definition}

This definition is very natural for cases where the oracle is implemented by a quantum algorithm that produces a valid answer with probability $\geq 1-\rho$.
In order to achieve our quantum speed-ups we will always assume access to the inverse $U^\dagger$ of the relational oracle as well, which is justified if $U$ comes from an efficiently implementable quantum algorithm.  

\subsection{Oracles for convex sets} \label{sec:oracles}

The five basic oracles for a convex set $K$ that we consider are as follows (in contrast with the original~\cite{grotschel1988GeomAlgAndConvOpt}, we allow some error probability $\rho$ in these oracles as in~\cite{lee2017ConvexOptWMemb}). Throughout we will assume that real vectors are represented with $\mathrm{polylog}(nR/(r\epsilon))$ bits of precision per coordinate. In particular, we assume that the input / output of the following oracles is represented this way.\footnote{Note that for weak oracles, where $\epsilon>0$, this is essentially without loss of generality, since the rounding errors can be incorporated into the error parameter of the oracle.}

\begin{definition}[Membership oracle $\mathrm{MEM}_{\epsilon,\rho}(K)$] \label{def:mem}
	Queried with a vector $y \in \R^n$, the oracle, with success probability $\geq 1-\rho$, correctly asserts one of the following 
	\begin{itemize}
		\item $y \in B(K,\epsilon)$, or 
		\item $y \not \in B(K,-\epsilon)$. 
	\end{itemize}
\end{definition}

\begin{definition}[Separation oracle $\mathrm{SEP}_{\epsilon,\rho}(K)$] \label{def:sep}
	Queried with a vector $y \in \R^n$, the oracle, with success probability at least $\geq 1-\rho$, correctly asserts one of the following
	\begin{itemize}
		\item $y \in B(K,\epsilon)$, or 
		\item $y \not \in B(K,-\epsilon)$, 
	\end{itemize}
	and in the second case it returns a unit vector $g \in \R^n$ such that $\langle g,x \rangle \leq \langle g,y \rangle +\epsilon$ for all $x \in B(K,-\epsilon)$.
\end{definition}

\begin{definition}[Optimization oracle $\mathrm{OPT}_{\epsilon,\rho}(K)$]
	Queried with a unit vector $c \in \R^n$, the oracle, with probability $\geq 1-\rho$, does one of the following:
	\begin{itemize}
		\item it returns a vector $y \in \R^n$ such that $y \in B(K,\epsilon)$ and $\ipc{c}{x} \leq \ipc{c}{y} + \epsilon$ for all $x \in B(K,-\epsilon)$, 
		\item or it correctly asserts that $B(K,-\epsilon)$ is empty. 
	\end{itemize}
\end{definition}
\noindent
Note that the above optimization oracle corresponds to \emph{maximizing} a linear function over a convex set; we could equally well state it for minimization.

\begin{definition}[Violation oracle $\mathrm{VIOL}_{\epsilon,\rho}(K)$]
	Queried with a unit vector $c \in \R^n$ and a real number~$\gamma$, the oracle, with probability $\geq 1-\rho$, does one of the following:
	\begin{itemize}
		\item it asserts that $\ipc{c}{x} \leq \gamma + \epsilon$ for all $x \in B(K,-\epsilon)$, 
		\item or it finds a vector $y \in B(K,\epsilon)$ such that $\ipc{c}{y} \geq \gamma - \epsilon$. 
	\end{itemize}
\end{definition}
\begin{definition}[Validity oracle $\mathrm{VAL}_{\epsilon,\rho}(K)$]
	Queried with a unit vector $c \in \R^n$ and a real number~$\gamma$, the oracle, with probability $\geq 1-\rho$, does one of the following:
	\begin{itemize}
		\item it asserts that $\ipc{c}{x} \leq \gamma + \epsilon$ for all $x \in B(K,-\epsilon)$, 
		\item or it asserts that $\ipc{c}{y} \geq \gamma - \epsilon$ for some $y \in B(K,\epsilon)$. 
	\end{itemize}
\end{definition}

\noindent If in the above definitions both $\eps$ and $\rho$ are equal to $0$, then we call the oracle \emph{strong}. If either is non-zero then we sometimes call it \emph{weak}.

\medskip
The above describes the classical oracles, and the quantum oracles are defined analogously, i.e., they are relational quantum oracles (see Definition~\ref{def:reloracle}), that use a binary representation for the input / output vectors.

When we discuss membership queries, we will always assume that we are given a small ball which lies inside the convex set. It is easy to see that without such a small ball one cannot obtain an optimization oracle using only $\mathrm{poly}(n)$ classical queries to a membership oracle (see, e.g.,~\cite[Sec.~4.1]{grotschel1988GeomAlgAndConvOpt} or the example below). As the following example shows, the same holds for quantum queries. 
We will use a reduction from a version of the well-studied \emph{search} problem:

\medskip
\emph{Given $z \in \{0,1\}^N$ such that $|z|=1$, find $b \in [N]$ such that $z_b = 1$.}
\medskip

\noindent It is not hard to see that if the access to $z$ is given via classical queries $i \mapsto z_i$, then $\Omega(N)$ queries are needed. It is well known~\cite{bennett1997QSearchLowerBound} that if we allow quantum queries, i.e., applications of the unitary $\ket{i}\ket{b} \mapsto \ket{i}\ket{z_i \oplus b}$, then $\Omega(\sqrt{N})$ queries are needed. Now let $N = 2^n$ and consider an input $z\in \{0,1\}^N$ to the search problem. Let $b \in \{0,1\}^n$ be the index such that $z_b = 1$. Consider maximizing the linear function $\langle e, z \rangle$ (where $e$ is the all-1 vector) over the set $K_z = \prod_{i=1}^n [b_i-1/2,b_i]$. Clearly the optimal solution to this convex optimization problem, even with a small constant additive error in the answer, gives the solution to the search problem. However, a membership query is essentially equivalent to querying a bit of~$z$ and therefore $\Omega(\sqrt{N}) = \Omega(2^{n/2})$ quantum queries to the membership oracle are needed for optimization. 

\section{Computing approximate subgradients of convex Lipschitz functions} \label{sec:approxsubgradient}

Here we show how to compute an approximate subgradient (at~$0$) of a convex Lipschitz function. That is, given a convex set $\domf$ such that $0 \in \mathrm{int}(C)$ and a convex function $f:\domf \rightarrow \R$, we show how to compute a vector $\tilde g \in \R^n$ such that $f(y) \geq f(0) + \langle \tilde g, y \rangle - a \nrm{y} - b$ for some real numbers $a,b>0$ that will be defined later (see  Lemma~\ref{lemma:hyperplane} and Lemma~\ref{lemma:quantHyperplane}). The idea of the classical algorithm given in the next section is to pick a point $z \in B_\infty(0,r_1)$ uniformly at random and use the finite difference $\nabla^{(r_2)}f(z)$ (defined below) as an approximate subgradient of $f$ at $0$; the radii $r_1$ and $r_2$ need to be chosen small to make the approximation good. This results in a slightly simplified version of the algorithm of Lee et al.~\cite{lee2017ConvexOptWMemb}.
In Section~\ref{sec:quantumgrad} we show how to improve on this classical algorithm on a quantum computer. 

\subsection{Classical approach}
In the discussion that follows we will use the following approximation of the gradient.
\begin{definition}[Finite-difference gradient approximation] \label{def:finitdif}
	For a function $f:\domf \rightarrow \R$, a real $r>0$, and a point $x\in \R^n$ such that $B_1(x,r)\subseteq \domf $, and $i\in[n]$, we define $$\nabla_i^{(r)}f(x):=\frac{f(x+r e_i)-f(x-r e_i)}{2r},$$ where $e_i\in\01^n$ is the vector that has a~1 only in its $i$th coordinate. 
	Similarly we define 
	\[
	\nabla^{(r)}f(x):=\left(\nabla_1^{(r)}f(x),\nabla_2^{(r)}f(x),\ldots,\nabla_n^{(r)}f(x)\right).
	\]
\end{definition}
We will also consider a similar approximation of the Laplacian (the trace of the Hessian) of a function.
\begin{definition}[Finite-difference Laplace approximation]
	For a function $f:\domf \rightarrow \R$, a real $r>0$, and a point $x\in \R^n$ such that $B_1(x,r)\subseteq \domf $, and $i\in[n]$, we define $$\Delta_i^{\!\!(r)}f(x):=\frac{f(x+r e_i)-2f(x) +f(x-r e_i)}{r^2}.$$
	Similarly 
	$$\Delta^{\!\!(r)}f(x):=\sum_{i=1}^{n}\Delta_i^{\!\!(r)}f(x).$$
\end{definition}
\noindent Note that for a convex function we have $\Delta_i^{\!\!(r)}f(x)\geq 0$ for all $x$ such that $B_1(x,r)\subseteq \domf$.

The next two lemmas will be needed in the proof of the main result of this section, Lemma~\ref{lemma:hyperplane}. In Lemma~\ref{lemma:expectation} we give an upper bound on the deviation $\nrm{g-\nabla^{(r_2)}f(z)}_1$ of a finite difference gradient approximation $\nabla^{(r_2)}f(z)$ from an actual subgradient $g$ at the point $z$, in terms of the finite difference Laplace approximation $\Delta^{\!\!(r_2)}f(z)$. Then, in Lemma~\ref{lemma:averageLaplaceBound} we show that in expectation (over the points of a small ball around $x$), the finite difference Laplace approximation is small. Together with Markov's inequality this gives us good control over the quality of a finite difference gradient approximation.

\begin{lemma}\label{lemma:expectation}
	If $r_2 >0$, $z\in\R^n$, and $f:B_1(z,r_2)\rightarrow \R$ is convex, then
	$$ \sup_{g\in\underline{\partial} f(z)}\nrm{g-\nabla^{(r_2)}f(z)}_1\leq \frac{r_2\Delta^{\!\!(r_2)}f(z)}{2}.$$
\end{lemma}
\begin{proof}
	Fix a $g\in\underline{\partial} f(z)$. For every $i\in[n]$, we have $f(z+r_2 e_i) \geq f(z) + \langle g, r_2 e_i\rangle = f(z) + r_2 g_i$, and, similarly, $f(z-r_2e_i) \geq f(z) - r_2 g_i$. Rearranging gives 
	\[
	\underbrace{\frac{f(z) - f(z-r_2 e_i)}{r_2}}_{:=A} \leq g_i \leq \underbrace{\frac{f(z+r_2e_i) - f(z)}{r_2}}_{:=B}. 
	\]
	Note that $|g_i-\frac{A+B}{2}| \leq \frac{B-A}{2}$ for any three real numbers $A \leq g_i \leq B$. Moreover, $\frac{A+B}{2} = \nabla_i^{(r_2)}f(z)$ and $B-A = r_2\Delta_i^{\!\!(r_2)}f(z)$, thus $\left|g_i-\nabla_i^{(r_2)}f(z)\right|\leq \frac{r_2\Delta_i^{\!\!(r_2)}f(z)}{2}$.
	Now we can finish the proof by summing this inequality over all $i\in[n]$.
\end{proof}

\begin{lemma}\label{lemma:averageLaplaceBound}
	If $0<r_2 \leq r_1 $, and $f:B_\infty(x,r_1+r_2)\rightarrow \R$ is convex and $L$-Lipschitz, then
	$$ \underset{z\in B_\infty(x,r_1)}{\E} \Delta^{\!\!(r_2)}f(z)\leq \frac{nL}{r_1}.$$
\end{lemma}
\begin{proof}
	Below we show that $\underset{z\in B_\infty(x,r_1)}{\E} \Delta_i^{\!\!(r_2)}f(z) \leq \frac{L}{r_1}$ for all $i\in[n]$, summing over $i$ then proves the lemma.  
	
	Let $h_i(z) := f(z-r_2e_i) - f(z)$; we have that
	\begin{align*}
	\underset{z\in B_\infty(x,r_1)}{\E} \Delta_i^{\!\!(r_2)}f(z) &= \frac{1}{(2r_1)^n} \int_{z \in B_\infty(x,r_1)} \frac{f(z+r_2 e_i) - 2f(z) + f(z-r_2e_i)}{r^2_2} \, dz\\
	&=\frac{1}{(2r_1)^n} \int_{\substack{z_j \in [x_j - r_1, x_j+r_1],\\ j \in [n], j \neq i}}  \int_{z_i \in [x_i -r_1, x_i+r_1]} \frac{f(z-r_2 e_i) - 2f(z) + f(z+r_2e_i)}{r^2_2} \, dz\\
	&=\frac{1}{(2r_1)^n} \int_{\substack{z_j \in [x_j - r_1, x_j+r_1],\\ j \in [n], j \neq i}}  \Big(\int_{z_i \in [x_i -r_1, x_i+r_1]} \frac{h_i(z)}{r^2_2} \, dz\\
	&\qquad \qquad \qquad \qquad \quad \qquad- \int_{z_i \in [x_i -r_1, x_i+r_1]} \frac{h_i(z+r_2e_i)}{r^2_2} \, dz \Big)\\
	&=\frac{1}{(2r_1)^n} \int_{\substack{z_j \in [x_j - r_1, x_j+r_1],\\ j \in [n], j \neq i}}  \Big(\int_{z_i \in [x_i -r_1, x_i+r_1]} \frac{h_i(z)}{r^2_2} \, dz\\
	&\qquad \qquad \qquad \qquad \quad \qquad- \int_{z_i \in [x_i -r_1+r_2, x_i+r_1+r_2]} \frac{h_i(z)}{r^2_2} \, dz \Big)\\
	&=\frac{1}{(2r_1)^n} \int_{\substack{z_j \in [x_j - r_1, x_j+r_1],\\ j \in [n], j \neq i}}  \Big(\int_{z_i \in [x_i -r_1, x_i-r_1+r_2]} \frac{h_i(z)}{r^2_2} \, dz\\
	&\qquad \qquad \qquad \qquad \quad \qquad- \int_{z_i \in [x_i +r_1, x_i+r_1+r_2]} \frac{h_i(z)}{r^2_2} \, dz \Big)\\
                                                               &\leq \frac{1}{(2r_1)^n} \int_{\substack{z_j \in [x_j - r_1, x_j+r_1],\\ j \in [n], j \neq i}} 2L \, dz\\
                                                               &= \frac{L}{r_1} .
	\end{align*}
The last inequality above follows from multiplying the upper bound $r_2L$ on $|h_i|$ with the length $r_2$ of the integration intervals.
\end{proof}
Note that the above lemma is stated and proved for continuous random variables, but the same proof holds if we have a uniform hypergrid over the same hypercube, providing a discrete version of the above result. In the discrete case, in order to get the same cancellations we need to assume that both $r_1$ and $r_2$ are integer multiples of the grid spacing.

We are now ready to prove the main result of this section. Informally, the next lemma proves that an approximate subgradient of a convex Lipschitz function $f$ at $0$ can be obtained by an algorithm that outputs $\nabla^{(r_2)} \tilde{f}(z)$ for a random $z$ close enough to~$0$, where $\tilde f$ is an approximate version of~$f$. In other words, this lemma gives us a classical algorithm to compute an approximate subgradient of $f$ using $2n$ classical queries to an approximate version of~$f$.

\begin{lemma}\label{lemma:hyperplane}
	Let $r_1>0$, $L>0$, $\rho\in(0,1/3]$, $\delta\in(0,r_1 \sqrt{n} L/\rho]$, then $r_2 := \sqrt{\frac{\delta r_1 \rho}{ \sqrt{n}L}}\leq r_1$. Suppose $f:\domf  \to \R$ is a convex function that is $L$-Lipschitz on $B_\infty(0,2r_1)$, and $\tilde{f}:B_\infty(0,2r_1) \to \R$ is such that $\nrm{\tilde{f}-f}_\infty\leq \delta$.
	Then for a uniformly random $z\in B_\infty(0,r_1)$, with probability at least $1-\rho$,
	\[
	f(y) \geq f(0) + \ipc{ \nabla^{(r_2)} \tilde{f}(z)}{y} - \frac{3n^{\frac{3}{4}}}{2}\sqrt{\frac{\delta L}{\rho r_1}}\nrm{y}  - 2L \sqrt{n} r_1 \qquad \text{for all } y \in C.
	\]
\end{lemma}
\begin{proof}
	Let $z \in B_\infty(0,r_1)$ and $g \in \underline{\partial} f(z)$. 
	Recall $\nrm{g}\leq L$ by Equation~(\ref{eq:subgradBound}).
	Then for all $y\in C$
	\begin{align*}
	f(y) &\geq f(z) + \langle g, y-z\rangle \\
	&=  f(z) + \langle g, y-z\rangle + \left( \ipc{ \nabla^{(r_2)} f(z)}{y} - \ipc{ \nabla^{(r_2)} f(z)}{y} \right) + \left( f(0) - f(0) \right)\\
	&=f(0) + \ipc{ \nabla^{(r_2)} f(z)}{y} + \langle g-\nabla^{(r_2)} f(z), y\rangle + (f(z)-f(0)) + \langle g,-z\rangle   \\
	&\geq f(0) + \ipc{ \nabla^{(r_2)} f(z)}{y} - \nrm{g-\nabla^{(r_2)} f(z)}_1 \nrm{y}_\infty - L\nrm{z} -  \nrm{g} \nrm{z} \\
	&\geq f(0) + \ipc{ \nabla^{(r_2)} f(z)}{y} - \nrm{g-\nabla^{(r_2)} f(z)}_1 \nrm{y}_\infty - L\sqrt{n}r_1  - L\sqrt{n}r_1\\
	&\geq f(0) + \ipc{ \nabla^{(r_2)} \tilde{f}(z)}{y} - \frac{\delta \sqrt{n}}{r_2}\nrm{y} - \nrm{g-\nabla^{(r_2)} f(z)}_1 \nrm{y}_\infty - 2L\sqrt{n}r_1.
	\end{align*}
	Note that in the last line we switched from $f$ to $\tilde{f}$, using that $\nabla^{(r_2)} f(z)$ and $\nabla^{(r_2)} \tilde{f}(z)$ differ by at most $\delta/r_2$ in each coordinate.
	Our choice of $r_2$ gives $\frac{\delta \sqrt{n}}{r_2}=n^{\frac{3}{4}}\sqrt{\frac{\delta L}{\rho r_1}}$ and by Lemma~\ref{lemma:expectation}--\ref{lemma:averageLaplaceBound} we have 
	$$ 
	\underset{z\in B_\infty(x,r_1)}{\E} \nrm{g-\nabla^{(r_2)}f(z)}_1\leq \frac{nLr_2}{2r_1}=\frac{n^{\frac{3}{4}}}{2}\sqrt{\frac{\delta L\rho}{r_1}}.
	$$
	By Markov's inequality we get that $\nrm{g-\nabla^{(r_2)}f(z)}_1\leq\frac{n^{\frac{3}{4}}}{2}\sqrt{\frac{\delta L}{\rho r_1}}$  with probability $\geq 1-\rho$ over the choice of~$z$. Plugging this bound on $\nrm{g-\nabla^{(r_2)}f(z)}_1$ into the above lower bound on $f(y)$ concludes the proof of the lemma.
\end{proof}

\subsection{Quantum improvements} \label{sec:quantumgrad}

In this section we show how to improve subgradient computation of convex functions via Jordan's quantum algorithm for gradient computation~\cite{jordan2005QuantGrad}. We use the formulation given by Gily\'{e}n et al.~\cite[Lemma~20]{gilyen2017OptQOptAlgGrad}, for which we first introduce the following definition.

\begin{definition}[Hyper-grid]\label{def:labelDefiniton}
	For $k\in \mathbb{N}$ we define the following discretization of the interval $(-1/2,1/2)$:
	\[
	G_k:=\left\{\frac{j}{2^k}-\frac{1}{2}+2^{-k-1} : j\in \{0,\ldots,2^k-1\}  \right\}~\subset~(-1/2,1/2).
	\]
	Similarly we define the $n$-dimensional hyper-grid $G^n_k$, which is the $n$-fold Cartesian product of $G_k$ with itself.
\end{definition}
\noindent Note that an element of $G_k^n$ can be represented using $n\times k$ (qu)bits. Basically, Jordan's algorithm just sets up a uniform superposition over all grid points, applies a ``phase query'' to~$f$, and then a quantum Fourier transform over each coordinate.

\begin{lemma}\textrm{\normalfont(Jordan's quantum gradient computation algorithm~\cite[Lemma~20]{gilyen2017OptQOptAlgGrad})}\label{lemma:genericJordan}\\
	Let $m\in \N$, $c\in \R$ and $g\in \R^n$ such that $\nrm{g}_\infty\leq 1/3$. If $h:G_m^n\to \R$ is such that 
	\begin{equation}\label{eq:closeFunctionApx}
	\left|h(x)-\ipc{g}{x} -c \right|\leq \frac{2^{-m}}{42\pi},
	\end{equation}
	for 99.9\% of the points $x\in G_m^n$,
	then using a single query to a phase oracle $\mathrm{O}\colon \ket{x}\mapsto e^{2\pi i 2^m h(x)}\ket{x}$ Jordan's gradient computation algorithm outputs a vector $v\in \R^n$ such that: 
	$$\Pr\left[|v_i-g_i|>\! 2^{2-m}\right]\leq 1/3 \quad \text{ for every  } i\in[n].$$
\end{lemma}

We now show that the above algorithm allows us to compute an approximate subgradient of a function $f$, even if we are only given standard oracle access to a function $\tilde f$ which is sufficiently close to $f$. In particular, we will assume we are given access to a standard unitary oracle of a function $\tilde{f}\colon G_m^n\to \R$ which satisfies $|\tilde{f}(x)-f(x)|\leq \delta$ for all $x \in G_m^n$. That is, we assume we are given access to a unitary $U$ acting as 
\begin{equation}\label{eq:unitaryEval}
U:\ket{x}\ket{0}\mapsto\ket{x}\ket{\tilde f(x)}
\end{equation}
Note that if we can classically efficiently evaluate~$\tilde f$, then it is well known that we can construct such a unitary as a small quantum circuit (see~\cite[Sec.~1.4.1]{nielsen2002QCQI}).

The main idea is that, using one application of $U$, a phase gate corresponding to the output register, and another application of $U^\dagger$ to uncompute the function value, we can implement a phase oracle for $\tilde f$.  Moreover, Equation~\eqref{eq:apxLinJordan} below  
will also hold for $\tilde f$, with a slightly worse right-hand side, since $f$ is close to $\tilde f$. 
A version of the following is proven in~\cite[Theorem~21]{gilyen2017OptQOptAlgGrad},  for completeness we sketch a proof.

\begin{corollary}[Gradient computation using approximate function evaluation]\label{cor:corollaryJordan}
	Let $\delta,B,r\in \R_+,c\in\R$, $\rho\in(0,1/3]$. Let $x_0,g\in\R^n$ with $\nrm{g}_\infty\leq \frac{B}{r}$. Let $m:=\left\lceil\log_2\left(\frac{B}{28\pi\delta}\right)\right\rceil$ and suppose $f:\left(x_0+r G_m^n\right)\to \R$ is such that
	\begin{equation}\label{eq:apxLinJordan}
	\left|f(x_0+r x)-\ipc{g}{rx}-c\right|\leq \delta
	\end{equation}
	for 99.9\% of the points $x\in G_m^n$, and
	we have access to a standard unitary oracle $U$, providing $\bigO{\log\left(\!\frac{B}{\delta}\!\right)}$-bit fixed-point binary approximations $\tilde{f}(z)$ s.t. $|\tilde{f}(z)-f(z)|\leq \delta$ for all $z\in \left(x_0+r G_m^n\right)$.
	Then we can compute a vector $\tilde{g}\in\R^n$ such that 
	$$\Pr\left[\,\nrm{\tilde{g}-g}_\infty>\! \frac{8\cdot 42\pi \delta}{r}\right]\leq \rho,$$
	with $\bigO{\!\log\!\big(\frac{n}{\rho}\big)\!}$ queries to $U$ and $U^\dagger$ 
	and gate complexity $\bigO{\!n\log\!\big(\frac{n}{\rho}\big)\!\log\!\big(\!\frac{B}{\delta}\!\big)\!\log\!\log\!\big(\frac{n}{\rho}\big)\!\log\!\log\!\big(\!\frac{B}{\delta}\!\big)\!}\!$.
\end{corollary}
\begin{proof}
	As described above the corollary, we first implement a phase oracle for $\tilde f$ and then we apply Jordan's gradient computation algorithm (Lemma~\ref{lemma:genericJordan}). 
	
	With a single query to $U$ and its inverse we can implement a phase oracle $\mathrm{O}$ that acts as $\mathrm{O}:\ket{x}\mapsto e^{2\pi i \frac{M}{3B} \tilde{f}(x_0+r x)}\ket{x}$, where $M:=\frac{3B}{84\pi\delta}$, and\footnote{We can assume without loss of generality that the upper bound $B$ is such that $M$ is a power of two.} $m:=\log_2(M)$. 
	Let $h(x):=\frac{\tilde{f}(x_0+r x)}{3B}$, then by~\eqref{eq:apxLinJordan} 99.9\% of the points $x\in G_m^n$ satisfy
	$\big|h(x)-\ipc{\frac{r}{3B}g}{x}-\frac{c}{3B}\big|\leq \frac{2\delta}{3B}=\frac{1}{42\pi M}$. Since $\nrm{\frac{r}{3B}g}_\infty\leq \frac{1}{3}$, by Lemma~\ref{lemma:genericJordan} we can compute a vector $v\in\R^n$ which is a coordinatewise $\frac{4}{M}$-approximator of $\frac{r}{3B}g$: for each $i\in[n]$ we have $\left|g_i - \frac{3B}{r}v_i\right|\leq \frac{12B}{r M}=\frac{8\cdot 42 \pi \delta}{r}$ with probability at least $\frac{2}{3}$.
	
	Note that the above success probability is per coordinate of~$g$. However, repeating the whole procedure $\mathcal{O}\big(\!\log(\frac{n}{\rho})\big)$ times and taking the median of the resulting vectors coordinatewise gives a gradient approximator $\tilde{g}$ with the desired approximation quality with probability at least $1-\rho$. For the proof of the gate complexity we refer\footnote{The correspondence with the parametrization of~\cite[Theorem~21]{gilyen2017OptQOptAlgGrad} is $\eps\leftrightarrow\frac{8\cdot 42\pi \delta}{r}$, $M\leftrightarrow \frac{B}{r}$.} to \cite[Theorem~21]{gilyen2017OptQOptAlgGrad} where the complexity of Jordan's algorithm is analyzed in detail.
\end{proof}

\paragraph{Remark.}
With essentially the same approach, the above corollary of Jordan's quantum gradient computation algorithm can also be proven in the setting where our access to an approximation of $f$ is not given by a standard quantum oracle but by a \emph{relational} quantum oracle, see Appendix~\ref{sec:app} for both the definition of this type of approximation to $f$ and a proof of this corollary. 

In terms of applications, we want to point out that if the membership oracle used in Section~\ref{sec:sep} comes from a deterministic algorithm, then we get a standard quantum oracle. Only when the membership oracle itself is relational (for example, when it is itself computed by a bounded-error quantum algorithm) do we need the more general setting of Appendix~\ref{sec:app}. 

\bigskip

\noindent
We would like to apply the above corollary to compute gradients of a convex Lipschitz function. To that end, the function needs to be sufficiently close to a linear function on a small region. Fortunately convex Lipschitz functions have this property. 
The following two lemmas ensure that Equation~\eqref{eq:apxLinJordan} holds.

\begin{lemma}\label{lemma:flat}
	Let $S\subseteq\R^n$ be such that $S=-S$, and let $\mathrm{conv}(S)$ denote the convex hull of~$S$. If $f: \mathrm{conv}(S)\to \R$ is a convex function, $f(0)=0$, and $\left|f(s)\right|\leq \delta$ for all $s\in S$, then 
	$$
	|f(s')|\leq \delta \text{ for all } s'\in\mathrm{conv}(S).
	$$
\end{lemma}
\begin{proof}
	Since $f$ is convex and $f(s)\leq \delta$ for all $s\in S$ we immediately get that $f(s')\leq \delta$ for all $s'\in\mathrm{conv}(S)$. Because $f(0)=0$ and $S=-S$, due to convexity we get that $f(s')\geq -f(-s')\geq -\delta$. 
\end{proof}

\begin{lemma}\label{lemma:almostLin}
	If $r_2>0$, $z\in\R^n$ and $f:B_1(z,r_2)\rightarrow \R$ is convex, then
	$$ \sup_{y\in B_1(0,r_2)}\left|f(z+y)-f(z)-\ipc{y}{\nabla^{(r_2)}f(z)}\right|\leq \frac{r_2^2 \Delta^{\!\!(r_2)}f(z)}{2}.$$
\end{lemma}
\begin{proof}
	Let $d(y):=f(z+y)-f(z)-\ipc{y}{\nabla^{(r_2)} f(z)}$ be the difference between $f(z+y)$ and its linear approximator. Let $S:=\left\{\pm r_2 e_i:i\in [n]\right\}$.
	It is easy to see that $d(0)=0$, $S=-S$, and $\mathrm{conv}(S)=B_1(0,r_2)$. Also, for all $s\in S$ we have $\left|d(s)\right|\leq r_2^2\Delta^{\!\!(r_2)}f(z)/2$:
	\begin{align*}
	d(\pm r_2 e_i)&= f(z\pm r_2 e_i)-f(z)-\ipc{\pm r_2 e_i}{\nabla^{(r_2)} f(z)}\\
	&= f(z\pm r_2 e_i)-f(z)\mp r_2 \nabla_i^{(r_2)} f(z)\\
	&= f(z\pm r_2 e_i)-f(z)\mp \frac{f(z+r_2 e_i)-f(z-r_2 e_i)}{2}\\
	&= \frac{f(z+r_2 e_i)-2f(z)+f(z-r_2 e_i)}{2}\\
	&= r_2^2 \Delta_i^{\!\!(r_2)}f(z)/2 
	\quad \leq r_2^2 \Delta^{\!\!(r_2)}f(z)/2.
	\end{align*}
	Therefore Lemma~\ref{lemma:flat} implies that $\sup_{y\in B_1(0,r_2)}|d(y)|\leq r_2^2 \Delta^{\!\!(r_2)}f(z)/2$.
\end{proof}

We can now state the main result of this section, the quantum analogue of Lemma~\ref{lemma:hyperplane}.

\begin{lemma} \label{lemma:quantHyperplane}
	Let $r_1>0$, $L>0$, $\rho\in(0,1/3]$, and suppose $\delta\in(0,r_1 n L/\rho]$. Suppose $f:\domf  \to \R$ is a convex function that is $L$-Lipschitz on $B_\infty(0,2r_1)$, and we have quantum query access\footnote{Using Corollary~\ref{cor:generalCorollaryJordan} instead of Corollary~\ref{cor:corollaryJordan} shows that a relational quantum oracle also suffices as input.\label{foot:relationalQuery}} to $\tilde{f}$, which is a $\delta$-approximate version of $f$, via a unitary $U$ over a (fine-enough) hypergrid of $B_\infty(0,2r_1)$.
	Then we can compute a $\tilde{g}\in\R^n$ using $\bigO{\log(n/\rho)}$ queries to $U$ and $U^\dagger$, such that with probability $\geq 1-\rho$, we have
	\[
	f(y) \geq  f(0) + \ipc{ \tilde{g}}{y} - 23^2 \sqrt{\frac{\delta n^3 L}{\rho r_1}}\nrm{y}_1- 2L\sqrt{n}r_1 \qquad \text{for all } y \in\domf
	\]
  and hence (by Cauchy-Schwarz)
  	\[
	f(y) \geq  f(0) + \ipc{ \tilde{g}}{y} - (23n)^2 \sqrt{\frac{\delta  L}{\rho r_1}}\nrm{y}- 2L\sqrt{n}r_1 \qquad \text{for all } y \in\domf.
	\]

\end{lemma}

\begin{proof}
    Let $r_2 := \sqrt{\frac{\delta r_1 \rho}{nL}}$ and note that $r_2\leq r_1$.
	The quantum algorithm works roughly as follows. It first picks a uniformly\footnote{A discrete quantum computer strictly speaking cannot do this, but (as noted after Lemma~\ref{lemma:averageLaplaceBound}) a uniformly random point from a fine-enough hypergrid suffices.} random $z \in B_\infty(0,r_1)$.
	Then it uses Jordan's quantum algorithm to compute an approximate gradient at $z$ by approximately evaluating $f$ in superposition over a discrete hypergrid of $B_\infty(z,r_2/n)$.
	This then yields an approximate subgradient of $f$ at~$0$.
	
	We now work out this rough idea.
	Since $B_\infty(z,r_2/n)\subseteq B_1(z,r_2)$, Lemma~\ref{lemma:almostLin} implies
	\begin{equation}\label{eq:functionFlat}
	\sup_{y\in B_\infty(0,r_2/n)}\left|f(z+y)-f(z)-\ipc{y}{\nabla^{(r_2)}f(z)}\right|\leq \frac{r_2^2 \Delta^{\!\!(r_2)}f(z)}{2}.
	\end{equation}
	Also as shown by Lemma~\ref{lemma:averageLaplaceBound} and Markov's inequality we have 
	\begin{equation}\label{eq:laplaceSmall}
	\Delta^{\!\!(r_2)}f(z)\leq \frac{2nL}{\rho r_1}
	\end{equation}
	with probability $\geq 1-\rho/2$ over the choice of~$z$.
	If $z$ is such that Equation~\eqref{eq:laplaceSmall} holds, then we get
	$$ \sup_{y\in B_\infty(0,r_2/n)}\left|f(z+y)-f(z)-\ipc{y}{\nabla^{(r_2)}f(z)}\right|\leq  \frac{nLr_2^2}{\rho r_1} = \delta.$$
	Now apply the quantum algorithm
	of Corollary~\ref{cor:corollaryJordan} with $r = 2r_2/n$, $c =  f(z)$, $g = \nabla^{(r_2)}f(z)$, and $B=Lr$. This uses $\bigO{\log(n/\rho)}$ queries to $U$ and $U^\dagger$, and with probability $\geq 1-\rho/2$ computes an approximate gradient $\tilde{g}$ such that 
	\begin{equation}\label{eq:gradWellApx}
	\nrm{\nabla^{(r_2)}f(z)-\tilde{g}}_\infty
	\leq\frac{8 \cdot 42 \pi n}{2 r_2} \cdot \delta
	=4\cdot 42\cdot \pi \sqrt{\frac{ \delta n^3 L}{\rho r_1}}.
	\end{equation}
	Also, if $z$ is such that Equation~\eqref{eq:laplaceSmall} holds, then by Lemma~\ref{lemma:expectation} we get that
	\begin{equation*}
	\sup_{g\in\underline{\partial} f(z)}\nrm{\nabla^{(r_2)}f(z)-g}_1\leq \frac{r_2\Delta^{\!\!(r_2)}f(z)}{2}\leq \frac{nLr_2}{\rho r_1} = \sqrt{\frac{\delta n L}{\rho r_1}},
	\end{equation*}
	and therefore by the triangle inequality and Equation~\eqref{eq:gradWellApx} we get that 
	\begin{align*}\label{eq:gradWellApx2}
	\sup_{g\in\underline{\partial} f(z)}\nrm{g-\tilde{g}}_\infty &\leq \sup_{g\in\underline{\partial} f(z)}\nrm{g-\nabla^{(r_2)}f(z)}_\infty + \nrm{\nabla^{(r_2)}f(z) - \tilde{g}}_\infty  \\
	&\leq \sup_{g\in\underline{\partial} f(z)}\nrm{g-\nabla^{(r_2)}f(z)}_1 + \nrm{\nabla^{(r_2)}f(z) - \tilde{g}}_\infty  \\
	&\leq \sqrt{\frac{\delta n L}{\rho r_1}} + 4\cdot 42\cdot \pi \sqrt{\frac{\delta n^3 L}{\rho r_1}} 
	\quad < 23^2 \sqrt{\frac{\delta n^3 L}{\rho r_1}}. 
	\end{align*}
	Thus with probability at least $1-\rho$, for all $y\in \domf$ and for all $g\in\underline{\partial} f(z)$ we have that
	\begin{align*}
	f(y) &\geq f(z) + \ipc{g}{y-z} \\
	&=f(0) + \ipc{ \tilde{g}}{y} + \ipc{g-\tilde{g}}{y} + (f(z)-f(0)) + \langle g,-z\rangle   \\
	&\geq f(0) + \ipc{ \tilde{g}}{y} - |\ipc{g-\tilde{g}}{y}| - L\nrm{z} -  \nrm{g} \nrm{z} \\
	&\geq f(0) + \ipc{ \tilde{g}}{y} - \nrm{g-\tilde{g}}_\infty\nrm{y}_1 - L\sqrt{n}r_1  - L\sqrt{n}r_1 \tag{by \eqref{eq:subgradBound}}\\
	&\geq f(0) + \ipc{ \tilde{g}}{y} -23^2 \sqrt{\frac{\delta n^3 L}{\rho r_1}}\nrm{y}_1- 2L\sqrt{n}r_1\\
	& \geq  f(0) + \ipc{ \tilde{g}}{y} - (23n)^2 \sqrt{\frac{\delta  L}{\rho r_1}}\nrm{y}- 2L\sqrt{n}r_1. \qedhere
	\end{align*}
\end{proof}

\section{Algorithms for separation using membership queries}
\label{sec:sep}

Let $K \subseteq \R^n$ be a convex set such that $B(0,r) \subseteq K \subseteq B(0,R)$. Given a membership oracle\footnote{For simplicity we assume throughout this section that the membership oracle succeeds with certainty (i.e., its error probability is~$0$). This is easy to justify: suppose we have a classical $T$-query algorithm, which uses $\mathrm{MEM}_{\epsilon,0}(K)$ queries and succeeds with probability at least $1-\rho$. If we are given access to a $\mathrm{MEM}_{\epsilon,\frac{1}{3}}(K)$ oracle instead, then we can create a $\mathrm{MEM}_{\epsilon,\frac{\rho}{T}}(K)$ oracle by $\bigO{\log(T/\rho)}$ queries to $\mathrm{MEM}_{\epsilon,\frac{1}{3}}(K)$ and taking the majority of the answers. Then running the original algorithm with $\mathrm{MEM}_{\epsilon,\frac{\rho}{T}}(K)$ will fail with probability at most $2\rho$. Therefore the assumption of a membership oracle with error probability~0 can be removed at the expense of only a small logarithmic overhead in the number of queries. A similar argument works for the quantum case.} $\mathrm{MEM}_{\epsilon,0}(K)$ as in Definition~\ref{def:mem}, we will construct a separation oracle $\mathrm{SEP}_{\eta,\rho}(K)$ as in Definition~\ref{def:sep}. Let $x$ be the point we want to separate from~$K$. We first make a membership query to $x$ itself, receiving answer $x\in B(K,\eps)$ or $x \not \in B(K,-\eps)$. Suppose $x \not \in B(K,-\eps)$, then we need to find a hyperplane that approximately separates $x$ from~$K$. Due to the rotational symmetry of the separation problem, for ease of notation we assume that $x=-\nrm{x}e_n$.\footnote{For the query complexity this is without loss of generality, since we can always apply a rotation to all the points such that this holds. If we instead consider the computational cost of our algorithm, then we have to take into account the cost of this rotation and its inverse. Note, however, that this rotation can always be written as the product of $n$ rotations on only $2$ coordinates, and hence can be applied in $\bOt{n}$ additional steps. These rotations can also be found in $\bOt{n}$ time via a greedy algorithm: first find a rotation on coordinates $n$ and $n-1$ that leaves coordinate $n-1$ zero, then similarly for coordinates $n-2$ and $n$, and so on.} We define $h:\R^{n-1}\rightarrow \R\cup\{\infty\}$ as
$$
h(y):=\inf_{(y,y_n)\in K} y_n,
$$
see also Figure~\ref{fig:defh}. Note that $h$ implicitly depends on $x$, since we have rotated the space such that $x = -\nrm{x}e_n$.
\begin{figure}[ht]
\centering
\includegraphics[trim=40 50 100 70,clip,width=.5\textwidth]{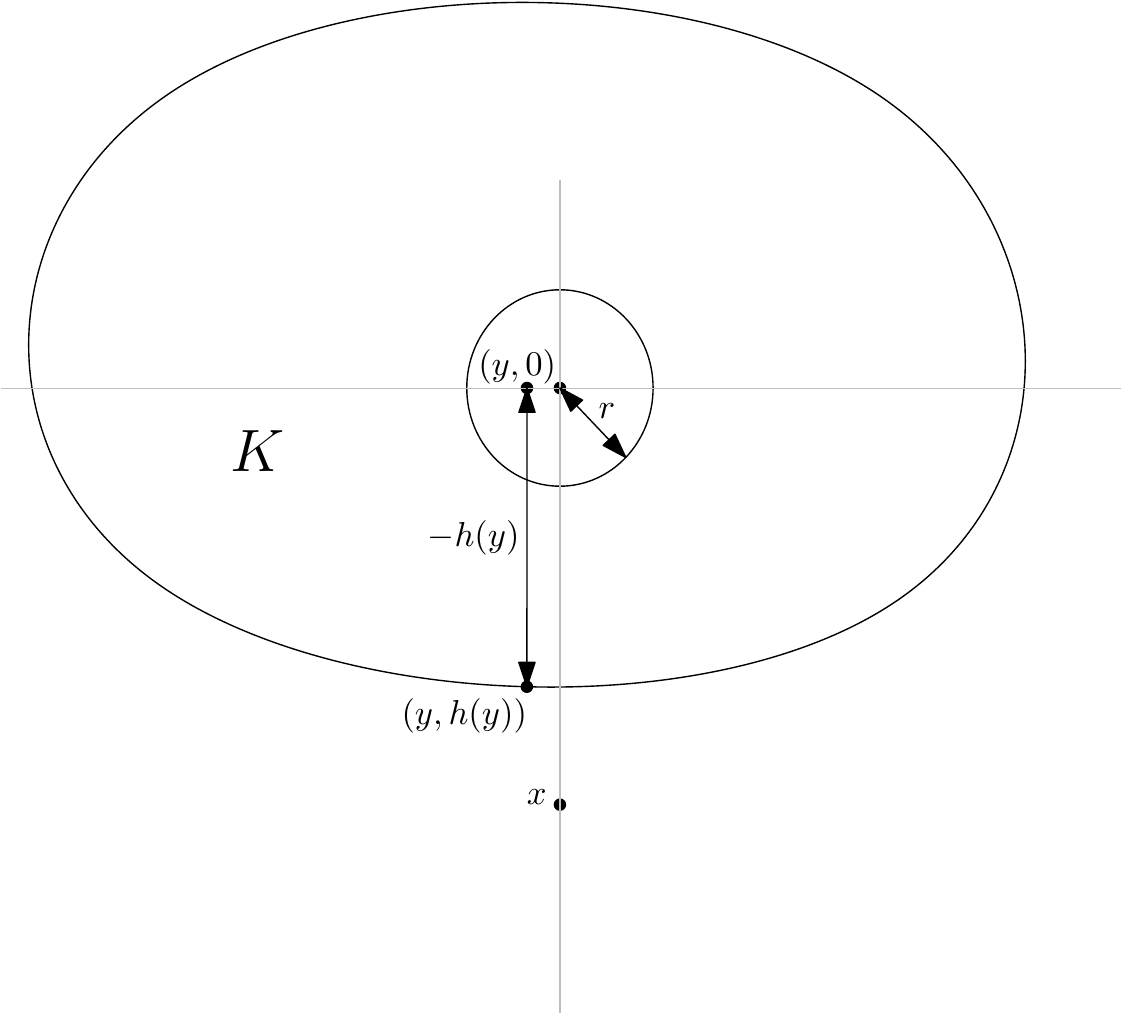}
\caption{Graphical example of the relation between $h(y)$ and the distance from $(y,0)$ to the border in the $-e_n$ direction.}
\label{fig:defh}
\end{figure}

Our $h$ is a bit different from the one used in~\cite{lee2017ConvexOptWMemb}, but we can show that it has many of the same properties.
Since $K$ is a convex set, $h$ is a convex function over $\R^{n-1}$. 
As we show below, the function $h$ is also Lipschitz (Lemma~\ref{lemma:hlipschitz}) and we can approximately compute its value using binary search with $\bOt{1}$ classical queries to a membership oracle (Lemma~\ref{lemma:approxEval}). 
Furthermore, an approximate subgradient of $h$ at~$0$ allows to construct a hyperplane approximately separating $x$ from~$K$ (Lemma~\ref{lemma:subToSep}). Combined with the results of Section~\ref{sec:approxsubgradient} this leads to the main results of this section, Theorems~\ref{thm:classsep} and~\ref{mainthrm:precise}, which show how to efficiently construct a separation oracle using respectively classical and quantum queries to a membership oracle.

Analogously to \cite[Lemma~12]{lee2017ConvexOptWMemb} we first show that our $h$ is Lipschitz.
\begin{lemma} \label{lemma:hlipschitz}
	For every $\delta\in(0,r)$, $h$ is $\frac{R}{r-\delta}$-Lipschitz on $B(0,\delta)\subseteq\R^{n-1}$, that is, we have 
	\[
	|h(y') - h(y)| \leq \frac{R}{r-\delta} \nrm{y'-y} \quad \text{for all } y,y' \in B(0,\delta).
	\] 
\end{lemma}
\begin{proof}
  Observe that for all $y\in B(0,r)$ we have $-R\leq h(y)\leq 0$, because $B(0,r) \subseteq K \subseteq B(0,R)$. Let $y,y'\in B(0,\delta)$ be arbitrary but distinct points. Due to symmetry it will suffice to show that $h(y')-h(y)\leq \frac{R}{r-\delta} \nrm{y'-y}$.

  We will restrict our attention to the line through $y$ and $y'$, i.e., the line given by $y+ \lambda z$ for $z:= \frac{y'-y}{\nrm{y'-y}}$. Define the point
  \[
    p:= y + \left(\nrm{y'-y} + \left(r-\delta\right)\right)z = y' + \left(r-\delta\right)z
  \]
  on this line and note that $p\in B(0,r)$. Since $y'$ lies between $y$ and $p$ on the line it is a convex combination of these two points. In particular, since $\nrm{p-y'}=r-\delta$, it is the convex combination
  \[
    y'  = \frac{\nrm{y'-y}}{\nrm{y'-y}+(r-\delta)} p + \frac{r-\delta}{\nrm{y'-y} + (r-\delta)} y.
  \]
  Due to convexity we have
  \[
    h(y') \leq \frac{\nrm{y'-y}}{\nrm{y'-y}+(r-\delta)} h(p) + \frac{r-\delta}{\nrm{y'-y} + (r-\delta)} h(y),
  \]
  which implies
  \[
    h(y') -h(y) \leq \frac{\nrm{y'-y}}{\nrm{y'-y}+(r-\delta)} \left(h(p) - h(y)\right) \leq \frac{\nrm{y'-y}}{r-\delta}R.\qedhere
  \]
\end{proof}
\anote{We could improve the above result to $\min(R,2\nrm{x})/(r-\delta)$.}

Now we show how to compute the value of $h$ using membership queries to $K$.

\begin{lemma}\label{lemma:approxEval}
	For all $y\in B\big(0,\frac{r}{2}\big)\subset\R^{n-1}$ we can compute a $\delta$-approximation of $h(y)$ with $\bigO{\log\left(\frac{R}{\delta}\right)}$ queries to a $\mathrm{MEM}_{\eps,0}(K)$ oracle, where $\eps \leq \frac{r}{3R} \delta$.
\end{lemma}

\begin{proof}
	Let $y\in B(0,\frac{r}{2})$, then $(y,h(y))$ is a boundary point of $K$ by the definition of~$h$. 
	Note that $h(y)\in[-R,-r/2]$. Our goal is to perform binary search over this interval to find a good approximation of~$h(y)$. If we had access to a perfect membership oracle, then this would be straightforward. However, since our membership oracle can give back a wrong answer when queried with a point that is $\eps$-close to the boundary of $K$, a more careful analysis is needed.

  Suppose
	$y_n \leq -\frac{r}{2}$ is our current guess for $h(y)$.  We first show that
	\begin{enumerate}[label=(\emph{\alph*})]
		\item\label{it:YES} if $(y,y_n) \in B(K,\eps)$,
		then $y_n\geq h(y)- \delta$, and
		\item\label{it:NO} if $(y,y_n) \not \in B(K,-\eps)$,
		then $y_n\leq h(y)+ \frac{2}{3}\delta$.
	\end{enumerate}
	
	For the proof of~\ref{it:YES} consider a $g\in \underline{\partial}h(y)$. Since $g$ is a subgradient we have that $h(z)\geq h(y)+\ipc{g}{z-y}$ for all $z\in\R^{n-1}$. Hence, for all $z \in \R^{n-1}$ and $z_n$ such that $(z,z_n) \in K$ we have 
	\[
	\ipc{\begin{pmatrix}-g \\ 1\end{pmatrix}}{\begin{pmatrix} y \\ h(y) \end{pmatrix}}\leq \ipc{\begin{pmatrix}-g \\ 1\end{pmatrix}}{\begin{pmatrix}z \\ h(z) \end{pmatrix}} \leq \ipc{\begin{pmatrix}-g \\ 1\end{pmatrix}}{\begin{pmatrix}z \\ z_n\end{pmatrix}}
	\]
	where the first inequality is a rewriting of the subgradient inequality and the second inequality uses that $z_n \geq h(z)$ since $(z,z_n) \in K$. 
	Since $(y,y_n)\in B(K,\eps)$ it follows from the above inequality that
	\[ 
	\ipc{\begin{pmatrix}-g \\ 1\end{pmatrix}}{\begin{pmatrix} y \\ y_n \end{pmatrix}}\geq\ipc{\begin{pmatrix}-g \\ 1\end{pmatrix}}{\begin{pmatrix} y \\ h(y) \end{pmatrix}}-\eps \nrm{\begin{pmatrix}-g \\ 1\end{pmatrix}} \geq \ipc{\begin{pmatrix}-g \\ 1\end{pmatrix}}{\begin{pmatrix} y \\ h(y) \end{pmatrix}}-\eps (\nrm{g}+1).
	\] 
	Lemma~\ref{lemma:hlipschitz} together with the argument of Equation~(\ref{eq:subgradBound}) implies that $\nrm{g}\leq \frac{2R}{r}$. 
	Since  
	$$
	\eps (\nrm{g}+1)\leq \eps \left(\frac{2R}{r}+1\right)\leq \eps \frac{3R}{r}\leq \delta,
	$$ 
	we obtain the inequality of~\ref{it:YES}.
	
	For~\ref{it:NO}, consider the convex set $C$ which is the convex hull of $B((y,0),r/2)$ and $(y,h(y))$.
	Note that $B(C,-\eps)$ is the convex hull of $B((y,0),r/2-\eps)$ and $\big(y,h(y)\left(1-\frac{2\eps}{r}\right)\big)$. Since $C\subseteq K$, we have $B(C,-\eps)\subseteq B(K,-\eps)$.
	Therefore $(y,y_n) \not \in B(K,-\eps)$ implies $(y,y_n)\notin B(C,-\eps)$, and
	\begin{equation*}
	y_n\leq h(y)\left(1-\frac{2\eps}{r}\right)=h(y)-\eps\frac{2 h(y)}{r}\leq h(y)+\eps\frac{2 R}{r}\leq h(y)+\frac{2}{3}\delta.
	\end{equation*}	
	
	Now we can analyze the binary search algorithm. By making $\bigO{\log\left(\frac{R}{\delta}\right)}$ $\mathrm{MEM}_{\eps,0}(K)$ queries to points of the form $(y,z)$, we can find a value $y_n\in[-R,-\frac{r}{2}]$ such that $(y, y_n) \in B(K,\eps)$ but $(y,y_n - \frac{\delta}{3}) \not \in B(K,-\eps)$.
	By \ref{it:YES}-\ref{it:NO} we get that $|h(y)-y_n|\leq \delta$.
\end{proof}

The following lemma shows how to convert an approximate subgradient of $h$ to a hyperplane that approximately separates $x$ from~$K$.

\begin{lemma}\label{lemma:subToSep}
	Suppose $-\nrm{x}e_n=x\notin B(K,-\eps)$, and $\tilde{g}\in\R^{n-1}$ is an approximate subgradient of $h$ at $0$, meaning that for some $a,b\in \R$ and for all $y \in \R^{n-1}$
	\[
	h(y) \geq h(0) + \langle \tilde g, y\rangle - a \nrm{y} - b,
	\]
	then $s:=\frac{(-\tilde{g},1)}{\nrm{(-\tilde{g},1)}}$ 
	satisfies $\ipc{s}{z}\geq \ipc{s}{x}-\frac{aR+b}{\nrm{(-\tilde{g},1)}} -\frac{2R}{r}\frac{\eps}{\nrm{(-\tilde{g},1)}}$ for all $z\in K$.
\end{lemma}
\begin{proof}
	Let us introduce the notation $z=(y,z_n)$ and $s':=(-\tilde{g},1)=\nrm{(-\tilde{g},1)}s$, then
	\begin{align*}
    \ipc{s'}{z} &=z_n-\ipc{\tilde{g}}{y}\\
	&\geq h(y)-\ipc{\tilde{g}}{y}\\
	&\geq h(0) - a \nrm{y} - b\\
	&\geq -\nrm{x} -\frac{2R}{r}\eps -aR -b\\
	&= \ipc{s'}{x} -aR -b  -\frac{2R}{r}\eps,
	\end{align*}
	where the last inequality used claim~\ref{it:NO} from the proof of Lemma~\ref{lemma:approxEval} with the point $(0,-\nrm{x})$ and $\delta = \frac{3 R}{r} \eps$.
\end{proof}

\noindent
We now construct a separation oracle using $\widetilde{\mathcal O}(n)$ classical queries to a membership oracle.
In particular, to construct an $\eta$-precise separation oracle, we require an $\eps$-precise membership oracle with
\[
\eps = \frac{\eta}{676} n^{-2}\left(\frac{r}{R}\right)^{\!3}\left(\frac{\eta}{R}\right)^{\!2}\rho
\]
The analogous result in \cite[Theorem~14]{lee2017ConvexOptWMemb} uses the stronger assumption\footnote{It seems that Lee et al.~\cite[Algorithm~1]{lee2017ConvexOptWMemb} did not take into account the change in precision analogous to our Lemma~\ref{lemma:approxEval}, therefore one would probably need to worsen their exponent of $\frac{r}{R}$ from~6 to~7.}
\[
\eps \approx \frac{\eta}{8\cdot 10^6}  n^{-\frac{7}{2}}\left(\frac{r}{R}\right)^{\!\!6}\left(\frac{\eta}{R}\right)^{\!\!2}\rho^3.
\]
Compared to this, our result scales better in terms of $n,\frac r R$ and $\rho$.

\begin{theorem} \label{thm:classsep}
	Let $K$ be a convex set satisfying $B(0,r) \subseteq K \subseteq B(0,R)$. For any $\eta\in(0,R]$ and $\rho\in(0,1/3]$, we can implement the oracle $\mathrm{SEP}_{\eta, \rho}(K)$ using $\bigO{n \log\left(\frac{n}{\rho}\frac{R}{\eta}\frac{R}{r}\right)}$ classical queries to a $\mathrm{MEM}_{\epsilon,0}(K)$ oracle, where $\eps\leq\eta (26n)^{-2}\left(\frac{r}{R}\right)^{\!3}\left(\frac{\eta}{R}\right)^{\!2}\rho$.
\end{theorem}

\begin{proof}
	Let $x \not \in B(K,-\eps)$ be the point we want to separate from~$K$.
	Let $\delta:=\eta\frac{n^{-2}}{9\cdot 24} \left(\frac{r}{R}\cdot\frac{\eta}{R}\right)^{\!2}\rho$,\linebreak 
	then $\eps\leq\frac{r}{3R}\delta$.
	By Lemma~\ref{lemma:hlipschitz} we know that $h$ is $\frac{2R}{r}$-Lipschitz on $B(0,r/2)$. By Lemma~\ref{lemma:approxEval} we can evaluate $h$ to within error $\delta$ using $\bigO{\log\left(\frac{R}{\delta}\right)}$ queries to a $\mathrm{MEM}_{\eps,0}(K)$ oracle.
	Let us choose $r_1:=\frac{r}{12\sqrt{n}}\frac{\eta}{R}$, then $r_1\sqrt{n}\leq \frac{r}{4}$, therefore $B_\infty(0,2r_1)\subseteq B(0,r/2)$. 
	Also note that $\delta\leq\frac{\eta}{6\rho}= \frac{2r_1\sqrt{n}R}{\rho r}$. 
	Hence by Lemma~\ref{lemma:hyperplane}, using $\bigO{n \log\left(\frac{R}{\delta}\right)}$ queries to a $\mathrm{MEM}_{\eps,0}(K)$ oracle, we can compute an approximate subgradient $\tilde{g}$
	such that with probability at least $1-\rho$ we have
	\[
	h(y) \geq h(0) + \langle \tilde g, y\rangle - \frac{3n^{\frac{3}{4}}}{2} \sqrt{\frac{\delta2R}{\rho r_1r}} \nrm{y} - \frac{4R}{r} \sqrt{n}r_1 \qquad \text{for all } y\in\R^{n-1}.
	\]
	Substituting the value of $r_1$ and $\delta$ we get $h(y) \geq h(0) + \langle \tilde g, y\rangle - \frac{\eta}{2R} \nrm{y} - \frac{\eta}{3} $, which by Lemma~\ref{lemma:subToSep} gives an $s$ such that $\ipc{s}{z}\geq \ipc{s}{x}-\frac{5}{6}\eta-\frac{2R}{r}\eps\geq\ipc{s}{x}-\eta$ for all $z\in K$  \qedhere
\end{proof}

Finally, we give a proof of our main result: we construct a separation oracle using $\widetilde{\mathcal O}(1)$ quantum queries to a membership oracle. 

\begin{theorem} \label{mainthrm:precise}
	Let $K$ be a convex set satisfying $B(0,r) \subseteq K \subseteq B(0,R)$. For any $\eta\in(0,R]$ and $\rho\in(0,1/3]$, we can implement the oracle $\mathrm{SEP}_{\eta, \rho}(K)$ using $\bigO{\log\big(\frac{n}{\rho}\big) \log\left(\frac{n}{\rho}\frac{R}{\eta}\frac{R}{r}\right)}$ quantum queries to a $\mathrm{MEM}_{\epsilon,0}(K)$ oracle (and its inverse), where $\eps\leq\eta (58n)^{-\frac{9}{2}}\left(\frac{r}{R}\right)^{\!3}\left(\frac{\eta}{R}\right)^{\!2} \rho $.
\end{theorem}

\begin{proof}
	Let $x \not \in B(K,-\eps)$ be the point that we want to separate from~$K$.
	Let us define $\delta:=\eta \frac{23^{-4}}{4\cdot 24}n^{\!-\frac{9}{2}}\left(\frac{r}{R}\cdot\frac{\eta}{R}\right)^{\!2} \rho$, then $\eps\leq\frac{r}{3R}\delta$.
	By Lemma~\ref{lemma:hlipschitz} we know that $h$ is $\frac{2R}{r}$-Lipschitz on $B(0,r/2)$. By Lemma~\ref{lemma:approxEval} we can evaluate $h$ to within error $\delta$ using $\bigO{\log\left(\frac{R}{\delta}\right)}$ queries to a $\mathrm{MEM}_{\eps,0}(K)$ oracle.
	Let us choose $r_1:=\frac{r}{12\sqrt{n}}\frac{\eta}{R}$, then $r_1\sqrt{n}\leq \frac{r}{4}$, therefore $B_\infty(0,2r_1)\subseteq B(0,r/2)$. 
	Also note that $\delta\leq\frac{\eta}{6\rho}= \frac{2r_1nR}{\rho r}$. 
	Hence by Lemma~\ref{lemma:quantHyperplane}, using $\bigO{\log\!\big(\frac{n}{\rho}\big)\!\log\left(\frac{R}{\delta}\right)\!}$ queries to a $\mathrm{MEM}_{\eps,0}(K)$ oracle,
	we can compute an approximate subgradient~$\tilde{g}$ such that with probability at least $1-\rho$ we have
	\[
	h(y) \geq h(0) + \langle \tilde g, y\rangle - (23n)^2\sqrt{\frac{2\delta R}{\rho r_1 r}}\nrm{y} - \frac{4R}{r} \sqrt{n}r_1 \qquad \text{for all } y\in\R^{n-1}.
	\]
	Substituting the value of $r_1$ and $\delta$ we get $h(y) \geq h(0) + \langle \tilde g, y\rangle - \frac{\eta}{2R} \nrm{y} - \frac{\eta}{3} $, which by Lemma~\ref{lemma:subToSep} gives an $s$ such that $\ipc{s}{z}\geq \ipc{s}{x}-\frac{5}{6}\eta-\frac{2R}{r}\eps\geq\ipc{s}{x}-\eta$ for all $z\in K$.  \qedhere
\end{proof}

\section{Lower bounds} \label{sec:lowerbounds}

For a convex set $K$ satisfying $B(0,r) \subseteq K \subseteq B(0,R)$, we have shown in Theorem~\ref{mainthrm:precise} that one can implement a $\mathrm{SEP}(K)$ oracle with $\bOt{1}$ quantum queries to a $\mathrm{MEM}(K)$ oracle if the membership oracle is sufficiently precise. In this section we first show that this is exponentially better than what can be achieved using classical access to a membership oracle. We also investigate how many queries to a membership/separation oracle are needed in order to implement an optimization oracle. Our results are as follows. 
\begin{itemize}
	\item We show that $\Omega(n)$ classical queries to a membership oracle are needed to implement a weak separation oracle.
	\item We show that $\Omega(n)$ classical (resp.~$\Omega(\sqrt{n})$ quantum) queries to a separation oracle are needed to implement a weak optimization oracle; even when we \emph{know an interior point} in the set.
	\item We show an $\Omega(n)$ lower bound on the number of classical and/or quantum queries to a separation oracle needed to optimize over the set when we \emph{do not know an interior point}.
\end{itemize}
In this section we will always assume that the input oracle is a strong oracle but the output oracle is allowed to be a weak oracle with error $\eps$. Furthermore, we will make sure that $R$, $1/r$, and $1/\eps$ are all upper bounded by a polynomial in~$n$. This guarantees that the lower bound is based on the dimension of the problem, not the required precision. 

\subsection{Classical lower bound on the number of MEM queries needed for SEP}

Here we show that a separation query can provide $\Omega(n)$ bits of information about the underlying convex set~$K$; since a classical membership query returns a 0 or a~1 and hence can give at most~1 bit of information,\footnote{This is not true for \emph{quantum} membership queries!} this theorem immediately implies a lower bound of $\Omega(n)$ on the number of classical membership queries needed to implement one separation query.

\begin{theorem} \label{thrm:lbinfo}
	Let $\epsilon \leq \frac{1}{48}$.
	There exist a set of $m = 2^{\Omega(n)}$ convex sets $K_1, \ldots, K_m$ and points $y,x_0 \in \R^n$ such that $B(x_0,1/3)\subseteq K_i \subseteq B(x_0,2\sqrt{n})$ for all $i\in[m]$, and such that the result of a classical query to $\mathrm{SEP}_{\eps,0}(K_i)$ with the point $y$ correctly identifies~$i$. 
\end{theorem}

\begin{proof}
	Let $h_1,\ldots,h_m \in \R^{n}$ be a set of $m = 2^{\Omega(n)}$ entrywise non-negative unit vectors such that $\ipc{h_i}{h_j}\leq 0.51$ for all distinct $i,j \in [m]$.\footnote{We can show that such a set of vectors exists as follows. Let $n=ck$ for sufficiently large constant $c$. Choose $m=2^k$ (which is $2^{\Omega(n)}$) uniformly random vectors $v_1,\ldots,v_m$ in $\01^n$. Note that the expected Hamming weight of one such vector is $n/2$, and the expected inner product between two vectors is $n/4$ (the inner product just counts for how many of the $n$ bit-positions both vectors have a~1). By a standard calculation (Chernoff bound plus a union bound), one can show that with high probability these $2^k$ vectors each have Hamming weight $\geq 0.495 n$, and the inner product between any two of them is $\leq 0.252 n$.
	Fix $2^k$ such vectors with these properties, and define $h_i:=v_i/\nrm{v_i}$. These are unit vectors with non-negative entries, and pairwise inner products $\ipc{h_i}{h_j}=\ipc{v_i}{v_j}/(\nrm{v_i}\nrm{v_j})\leq 0.252 n/(0.495 n)< 0.51$.}
	
	Now pick an $i \in [m]$ and define $\hat{K}_i:= \{x : \ipc{ h_i}{ x}  \leq 0\} \cap B(0,\sqrt{n})$ and $K_i := B(\hat K_i, \epsilon)$. Then $\hat K_i = B(K_i,-\epsilon)$.  Note that for $x_0 = -e / 3$ we have $B(x_0,1/3)\subseteq K_i \subseteq B(x_0,2\sqrt{n})$.
	We claim that a query to $\mathrm{SEP}_{\eps,0}(K_i)$ with the point $y=3\eps e\in\R^{n}$ will identify~$h_i$. First note that $y\not \in B(K_i,\epsilon)$, since $\hat{K}_i$ does not contain any entrywise positive vectors and $y$ has distance at least $3\eps$ from all vectors that have at least one non-positive entry.
	Hence a separation query with~$y$ must return a unit vector $g$ that describes a valid separating hyperplane for $K_i$. 
	
	On the other, if $g$ describes a valid separating hyperplane for $K_j$, then
	\begin{equation}
	\forall x \in \hat K_j\colon\ipc{g}{x} \leq \ipc{g}{y}+\eps \leq \nrm{g}\cdot\nrm{y}+\eps\leq (3\sqrt{n}+1)\eps\leq 4\sqrt{n}\eps. \label{eq:ipbound}
	\end{equation}
	Now consider the specific point~$x$ that is the projection of $g$ onto $h_j^{\bot}$ (the hyperplane orthogonal to $h_j$) scaled by a factor $\sqrt{n}$, i.e., $x=\sqrt{n}\left(g-\ipc{g}{h_j}h_j\right)$. Since $\ipc{h_j}{x}=0$ and $\nrm{x}\leq \sqrt{n}$, we have $x\in\hat{K}_j$. Choosing this $x$ in \eqref{eq:ipbound} gives the following inequality
	\[
	\sqrt{n} (1-\ipc{g}{h_j}^2)=\ipc{g}{x} \leq 4\sqrt{n}\eps.
	\]
	Hence \eqref{eq:ipbound} implies $|\ipc{g}{h_j}| \geq \sqrt{1-4\eps}\geq \sqrt{\frac{11}{12}}\geq \frac{19}{20}$. 
	
	Since \eqref{eq:ipbound} holds for $j=i$, it follows that at least one of the two vectors $g-h_i$ and $g+h_i$ has length at most $\sqrt{2(1-|\ipc{g}{h_i}|^2)}\leq \sqrt{8\epsilon}$; assume the former for simplicity. If \eqref{eq:ipbound} would also hold for $j\neq i$, then we would get a contradiction:
	$$
	\frac{19}{20}\leq|\ipc{g}{h_j}|\leq |\ipc{g-h_i}{h_j}|+|\ipc{h_i}{h_j}|\leq \sqrt{8\epsilon} + 0.51 < \frac{19}{20}.
	$$
	Hence $g$ uniquely identifies~$h_i$.
\end{proof}    

\subsection{Lower bound on number of SEP queries for OPT (given an interior point)}

We now consider lower bounding the number of quantum queries to a separation oracle needed to do optimization. In fact, we prove a lower bound on the number of separation queries needed for validity, which implies the same bound on optimization. We will use a reduction from a version\footnote{Note that this is a slightly different version from the one used in Section~\ref{sec:oracles}.} of the well-studied \emph{search} problem:

\medskip
\emph{Given $z \in \{0,1\}^n$ such that either $|z|=0$ or $|z|=1$, decide which of the two holds.}
\medskip

\noindent It is not hard to see that if the access to $z$ is given via classical queries, then $\Omega(n)$ queries are needed. It is well known~\cite{bennett1997QSearchLowerBound} that if we allow quantum queries, then $\Omega(\sqrt{n})$ queries are needed (i.e., Grover's quantum search algorithm~\cite{grover1996QSearch} is optimal). We use this problem to show that there exist convex sets for which it is hard to construct a weak validity oracle, given a strong separation oracle. Since a separation oracle can be used as a membership oracle, this gives the same hardness result for constructing a weak validity oracle from a strong membership oracle.  

\begin{theorem} 
	Let $0< \rho \leq 1/3$. Let $\mathcal A$ be an algorithm that implements a $\mathrm{VAL}_{(5n)^{-1},\rho}(K)$ oracle for every convex set $K$ (with $B(x_0,r) \subseteq K \subseteq B(x_0,R)$) using only queries to a $\mathrm{SEP}_{0,0}(K)$ oracle, and unitaries that are independent of~$K$. Then the following statements are true, even when we restrict to convex sets $K$ with $r = 1/3$ and $R = 2 \sqrt{n}$:
	\begin{itemize}
		\item if the queries to $\mathrm{SEP}_{0,0}(K)$ are classical, then the algorithm uses $\Omega(n)$ queries. 
		\item  if the queries to $\mathrm{SEP}_{0,0}(K)$ are quantum, then the algorithm uses $\Omega(\sqrt{n})$ queries. 
	\end{itemize}
\end{theorem}

\begin{proof}
	Let $z \in \{0,1\}^n$ have Hamming weight $|z| = 0$ or $|z|=1$. We construct a set $K_z$ in such a way that solving the weak validity problem solves the search problem for $z$, while separation queries for~$K_z$ can be answered using a single query to $z$. The known classical and quantum lower bounds on the search problem then imply the two claims of the theorem, respectively.
	
	Define $K_z := \prod_{i=1}^n [-1,z_i]$. 	Observe that if we set $x_0 = (-1/2,\dots,-1/2)$, then $B(x_0,\frac{1}{3}) \subseteq K_z \subseteq B(x_0,2\sqrt{n})$.

  We first show how to implement a strong separation oracle using a single query to $z$. Suppose the input is the point $y$. The strong separation oracle works as follows:
	\begin{enumerate}
		\item If $y \in [-1,0]^n$, then return the statement that $y \in B(K_z,0) = K_z$.
		\item If $y\not \in [-1,1]^n$, then return a hyperplane that separates $y$ from $[-1,1]^n$ (and hence from $K_z$).
		\item Otherwise, let $i$ be such that $y_i>0$. Query $z_i$. 
		\begin{enumerate}
			\item If $z_i = 1$ and $i$ is the only index such that $y_i>0$, then return that $y \in B(K_z,0) = K_z$. 
			\item If $z_i = 1$ and there is a $j\neq i$ such that $y_j>0$, return the separating hyperplane corresponding to $x_j \leq y_j$.
			\item If $z_i = 0$, then return the separating hyperplane $x_i \leq y_i$. 
		\end{enumerate}
	\end{enumerate}
	
We show that a validity query over $K_z$ with the direction $c = \frac{1}{\sqrt{n}} \left(1,\ldots, 1\right) \in \R^n$, value $\gamma=\frac{1}{2\sqrt{n}}$ and error $\eps=\frac{1}{5n}$ solves the search problem:
	\begin{itemize}
\item If $|z|=0$, then for all points $x\in K_0$ we have $\langle c,x \rangle \leq 0$. Thus, for all points $x \in B\left(K_0,\eps\right)$ we have $\langle c,x \rangle \leq \eps < \gamma - \eps$. Hence the validity oracle will have to return that $\langle c,x \rangle  \leq \gamma + \eps$ holds for all $x \in B\left(K_0,-\eps\right)$, since the other possible output is not true.

\item If $|z| = 1$, then the point $z \in K_z$ satisfies $\langle z, c \rangle= \frac{1}{\sqrt{n}}$ and therefore $x = z - \eps e \in B\left(K_z,-\eps\right)$ satisfies $\langle c,x \rangle = \frac{1}{\sqrt{n}} - \sqrt{n}\eps > \gamma+\eps$.
 Hence the validity oracle will have to return that $\langle c,x \rangle \geq \gamma - \eps$ holds for some $x \in B\left(K_z,\eps\right)$, since the other possible output is not true.

	\end{itemize}
\end{proof}

\subsection{Lower bound on number of SEP queries for OPT (without interior point)}

We now lower bound the number of quantum queries to a separation oracle needed to solve the optimization problem, if our algorithm does not already know an interior point of~$K$.
In fact we prove a lower bound on finding a point close to $K$ using separation queries, which implies the lower bound on the number of separation queries needed for optimization since $\opt$ returns a point close to the set~$K$.

We  prove our lower bound by a reduction to the problem of learning $z$ with \emph{first-difference queries}. Here one needs to find an initially unknown $n$-bit binary string $z$ via a guessing game. For a given guess $g\in\{0,1\}^n$ a query returns the first index in $[n]$ for which the binary strings $z$ and $g$ differ (or it returns $n+1$ if $z=g$). The goal is to recover $z$ with as few guesses as possible. First we prove an $\Omega(n)$ quantum query lower bound for this problem.\footnote{Note that this is a strengthening of the $\Omega(n)$ quantum query lower bound for binary search on a space of size $2^n$ by Ambainis~\cite{ambainis1999LowerBoundForOrderedSearch}, since first-difference queries are at least as strong as the queries one makes in binary search.}

\begin{theorem}[Quantum lower bound for learning $z$ with first-difference queries] \label{thm:lbfdq}
	Let $z\in\{0,1\}^n$ be an unknown string accessible by an oracle acting as $O_z\ket{g,b} = \ket{g,b\oplus f(g,z)}$, where $f(g,z)$ is the first index for which $z$ and $g$ differ, more precisely $f(g,z)=\min\{i\in [n]: g_i\neq z_i\}$ if $g\neq z$ and $f(g,z) = n+1$ otherwise. Then every quantum algorithm that outputs $z$ with high probability uses at least $\Omega(n)$ queries to $O_z$.
\end{theorem}

\begin{proof}
	We will use the general adversary bound~\cite{hoyer2007NegativeAdv}. For this problem, we call $\Gamma\in \mathbb{R}^{2^n\times 2^n}$ an \emph{adversary matrix} if it is a non-zero matrix with zero diagonal whose rows and columns are indexed by all~$z\in\{0,1\}^n$.
	For $g\in \{0,1\}^n$ let us define $\Delta_g\in\{0,1\}^{2^n\times 2^n}$ such that the $[z,z']$ entry of $\Delta_g$ is $0$ if and only if $f(g,z)= f(g,z')$. 
	The general adversary bound tells us that for any adversary matrix $\Gamma$, the quantum query complexity of our problem is 
	\begin{equation}\label{eq:adversaryBound}
	\Omega\left(\frac{\nrm{\Gamma}}{\max_{g\in\{0,1\}^n} \nrm{\Gamma \circ \Delta_g}}\right),
	\end{equation}
	where ``$\circ$'' denotes the Hadamard product and $\nrm{\cdot}$ the operator norm.
	
	We claim that Equation \eqref{eq:adversaryBound} gives a lower bound of $\Omega(n)$ for the adversary matrix $\Gamma$ defined as 
	\[
	\Gamma[z,z'] = \begin{cases}
	2^{f(z,z')} & \mbox{if }z\neq z'\\
	0 & \mbox{if }z=z'
	\end{cases}
	\]
	It is easy to see that $\Gamma$ is indeed an adversary matrix since it is zero on the diagonal and non-zero everywhere else. Furthermore, the all-one vector $e$ is an eigenvector of $\Gamma$ with eigenvalue $n2^n$:
	\[
	(\Gamma e)_z = \sum_{z'\in \{0,1\}^n} \Gamma[z,z'] = \sum_{d=1}^n 2^{d} \cdot |\{z' \in \{0,1\}^n \, : \, f(z,z') = d \}| = \sum_{d=1}^n 2^{d} 2^{n-d} = n2^n.
	\]
	So $\Gamma e = n2^n e$ and hence $\nrm{\Gamma}\geq n2^n$.
	
	From the definition of $\Delta_g$ it follows that
	\begin{align*}
	(\Gamma\circ\Delta_g)[z,z'] &= 
	2^{f(z,z')} \chi_{[f(g,z)\neq f(g,z')]},
	\end{align*}
	where $\chi_{[f(g,z)\neq f(g,z')]}$ stands for the indicator function of the condition $f(g,z)\neq f(g,z')$.
	Let $\Gamma_g := \Gamma\circ\Delta_g$. We will show an upper bound on $\nrm{\Gamma_g}$. We decompose $\Gamma_g$ in an ``upper-triangular'' and a ``lower-triangular'' part:
	\begin{align}
	\Gamma_g^U[z,z']  &:= 2^{f(z,z')}\chi_{[f(g,z)<f(g,z')]} = 2^{f(g,z)}\chi_{[f(g,z)<f(g,z')]},\label{eq:upperTriangleDef}\\
	\Gamma_g^L[z,z']  &:= 2^{f(z,z')}\chi_{[f(g,z')<f(g,z)]} = 2^{f(g,z')}\chi_{[f(g,z')<f(g,z)]}.\nonumber
	\end{align}
	So $\Gamma_g = \Gamma_g^U+\Gamma_g^L$ and $\Gamma_g^U = (\Gamma_g^{L})^T$. Hence by the triangle inequality we have
	\begin{equation}\label{eq:triangleGamma}
	\nrm{\Gamma_g} \leq \nrm{\Gamma_g^U}+\nrm{\Gamma_g^L} = 2\nrm{\Gamma_g^U}.
	\end{equation}
	It thus suffices to upper bound $\nrm{\Gamma_g^U}$. Notice that as \eqref{eq:upperTriangleDef} shows, $\Gamma_g^U[z,z']$ only depends on the values $f(g,z)$, $f(g,z')$.
	Since the range of $f(g,\,\cdot\,)$ is $[n+1]$, we can think of $\Gamma_g^U$ as an $(n+1)\times (n+1)$ block-matrix, where the blocks are determined by the values of $f(g,z)$ and $f(g,z')$, and within a block all matrix elements are the same.
	Also observe that for all $k\in[n]$ there are $2^{n-k}$ bitstrings $y\in\{0,1\}^n$ such that $f(g,y) = k$, which tells us the sizes of the blocks are $2^{n-k}\times 2^{n-k}$. Motivated by these observations we define an orthonormal set of vectors in $\R^{2^n}$ by $v_{n+1}:=e_g$, and for all $k\in[n]$
	\begin{align*}
	v_k:=\sum_{y: f(g,y)=k} \frac{e_y}{\sqrt{2^{n-k}}}.
	\end{align*}
	Since the row and column spaces of $\Gamma_g^U$ are spanned by $\{v_k:k\in[n+1]\}$, we can reduce $\Gamma_g^U$ to an $(n+1)\times (n+1)$-dimensional matrix $G$:
	\[
	\Gamma_g^U
	=\left(\sum_{k=1}^{n+1} v_k  v_k^T\right)\Gamma_g^U\left(\sum_{\ell=1}^{n+1} v_\ell  v_\ell^T\right)
	=\left(\sum_{k=1}^{n+1} v_k  e_k^T\right)\underbrace{\left(\sum_{k=1}^{n+1} e_k  v_k^T\right)\Gamma_g^U\left(\sum_{\ell=1}^{n+1} v_\ell  e_\ell^T\right)}_{G:=}\left(\sum_{\ell=1}^{n+1} e_\ell  v_\ell^T\right).
	\]
	It follows from the above identity, together with the orthonormality of $\{v_1,\ldots,v_n,v_{n+1}\}$, that 
	\begin{equation}\label{eq:triangleGamma2}
	\nrm{\Gamma_g^U}=\nrm{\left(\sum_{k=1}^{n+1} e_k  v_k^T\right)\Gamma_g^U\left(\sum_{\ell=1}^{n+1} v_\ell  e_\ell^T\right)}=\nrm{G}.
	\end{equation}
	$G\in\mathbb{R}^{(n+1)\times (n+1)}$ is a strictly upper-triangular matrix, with the following entries for $k,\ell\in [n]$:
	\begin{align*}
	G[k,\ell] &= v_k^T\Gamma_g^U v_\ell\\
	&= \left(\sum_{z: f(g,z)=k} \frac{e_z^T}{ \sqrt{2^{n-k}}}\right) \Gamma_g^U \left(\sum_{z': f(g,z')=\ell} \frac{e_{z'}}{ \sqrt{2^{n-\ell}}}\right)\\
	&= \frac{2^{\frac{k+\ell}{2}}}{2^{n}} \left(\sum_{z: f(g,z)=k} e_z^T\right) \Gamma_g^U \left(\sum_{z': f(g,z')=\ell} e_{z'}\right)\\
	&= \frac{2^{\frac{k+\ell}{2}}}{2^{n}} \sum_{z: f(g,z)=k} \sum_{z': f(g,z')=\ell}  \Gamma_g^U[z,z']\\	
	\end{align*}
	By Equation~\eqref{eq:upperTriangleDef} this is further equal to 
	\begin{align*}		
	G[k,\ell]  &= \frac{2^{\frac{k+\ell}{2}}}{2^{n}} \sum_{z: f(g,z)=k} \sum_{z': f(g,z')=\ell}  2^{k}\chi_{[k<\ell]}\\
	&= \frac{2^{\frac{k+\ell}{2}}}{2^{n}} 2^{n-k}  2^{n-\ell} 2^{k}\chi_{[k<\ell]}\\
	&= 2^{n-\frac{\ell-k}{2}}\chi_{[k<\ell]}.
	\end{align*}
	Similarly for $\ell=n+1$ we get that $G[k,\ell] = \sqrt{2}\, 2^{n-\frac{\ell-k}{2}}\chi_{[k<\ell]}$ for all $k\in [n+1]$.
	For each $d\in [n]$ define $G_d \in\mathbb{R}^{(n+1)\times (n+1)}$ such that $G_d[k,\ell]=G[k,\ell]\chi_{[d=\ell-k]}$. 
	This $G_d$ is only non-zero on one non-main diagonal (namely the $(k,\ell)$-entries where $d=\ell-k$), and its non-zero entries are all upper bounded by~$\sqrt{2}\,2^n 2^{-\frac{d}{2}}$.
	We have $G=\sum_{d=1}^{n} G_d$ and therefore
	\begin{equation}\label{eq:offdiagonalTriangle}
	\nrm{G}\leq \sum_{d=1}^{n} \nrm{G_d}\leq  \sum_{d=1}^{n}\sqrt{2}\,2^n 2^{-\frac{d}{2}}=2^n\sum_{d=0}^{n-1}(\sqrt{2})^{-d}\leq \frac{2^n}{1-1/\sqrt{2}}\leq 2^{n+2}.
	\end{equation}
	Inequalities \eqref{eq:triangleGamma}-\eqref{eq:offdiagonalTriangle} give that $\nrm{\Gamma_g} \leq 2^{n+3}$ and hence \eqref{eq:adversaryBound} yields a lower bound of $\Omega\left(\frac{n2^n}{2^{n+3}}\right)=\Omega(n)$ on the number of quantum queries to $O_z$ needed to learn $z$.	
\end{proof}

\begin{theorem}
	Finding a point in $B_\infty(K,1/7)$ for an unknown convex set $K$ such that $K\subseteq B_\infty(0,2) \subseteq \R^n$ requires $\Omega(n)$ quantum queries to a separation oracle $\mathrm{SEP}_{0,0}(K)$, even if we are promised there exists some unknown $x \in \R^n$ such that $B_\infty(x,1/3)\subseteq K$.
\end{theorem}

\begin{proof}
	We will prove an $\Omega(n)$ quantum query lower bound for this problem by a reduction from learning with first-difference queries.
	Let $z\in\{0,1\}^n$ be an unknown binary string, and let us define $K_z := B_{\infty}(z,1/3) \subset \mathbb{R}^n$ as a small box around the corner of the hypercube corresponding to $z$.
	Then clearly  $K_z \subset B_\infty(0,2)$, and finding a point close enough to $K_z$ is enough to recover $z$.
	
	We can easily reduce a separation oracle query to a first-difference query to~$z$, as follows. Suppose $y$ is the vector for which we need to answer a SEP query:
	\begin{enumerate}
		\item If $y$ is outside $[-1/3,4/3]^n$, then output a  hyperplane separating $y$ from $[-1/3,4/3]^n$.
		\item If $y$ is in $[-1/3,4/3]^n$, then let $g$ be the nearest corner of the hypercube.
		\item Let $i$ be the result of a first-difference query to $z$ with $g$.
		\begin{enumerate}
			\item If $i=n+1$, indicating that $z=g$, then we know $K_z$ exactly, so we can find a separating hyperplane or conclude that $y\in K_z$.
			\item If $z\neq g$, then return $e_i$ if $g_i=1$, and $-e_i$ if $g_i=0$.
		\end{enumerate}
	\end{enumerate}
	
	Hence our $\Omega(n)$ quantum lower bound on learning~$z$ with first-difference queries implies an $\Omega(n)$ lower bound on the number of quantum queries to a separation oracle needed for finding a point close to a convex set.
\end{proof}

Since optimization over a set $K$ gives a point close to the set $K$, this also implies a lower bound on the number of separation queries needed for optimization.
This theorem is tight up to logarithmic factors, since it is known that $\bOt{n}$ classical separation queries suffice for optimization, even without knowing a point in the convex set~\cite{lee2015FasterCuttingPlaneConvexOpt}.
Finally we remark that, due to our improved algorithm for optimization using validity queries (by combining Section~\ref{sec:corollaries} with Theorem~\ref{mainthrm:precise}), this also gives an $\widetilde{\Omega}(n)$ lower bound on the number of separation queries needed to implement validity.\footnote{It is easy to modify Theorem~\ref{thm:lbfdq} to prove a lower bound on computing the majority function applied to $z$, which would imply an $\Omega(n)$ lower bound on the number of separation queries needed to implement a validity oracle, without the log factors.}

\section{Consequences of convex polarity} \label{sec:corollaries}
Here we justify the central symmetry of Figure~\ref{fig:results} using the results of Gr{\"o}tschel, Lov\'asz, and Schrijver~\cite[Section~4.4]{grotschel1988GeomAlgAndConvOpt}. We first need to recall the definition and some basic properties of the polar $K^*$ of a set $K \subseteq \R^n$. This is the closed convex set defined as follows:
\[
K^* = \{y \in \R^n: \langle y, x \rangle \leq 1 \text{ for all } x \in K\}.
\]
It is straightforward to verify that if $B(0,r) \subseteq K \subseteq B(0,R)$, then $B(0,1/R) \subseteq K^* \subseteq B(0,1/r)$, moreover $(K^*)^* = K$ for closed convex sets.\footnote{Note that $K^*$ is a dual representation of the convex set $K$. Each point in $K^*$ corresponds to a (normalized) valid inequality for $K$. This duality is not to be confused with Lagrangian duality.} For the remainder of this section we assume that $K$ is a closed convex set such that $B(0,r) \subseteq K \subseteq B(0,R)$.

We will observe that for the polar $K^*$  of a set $K$ the following holds:
\begin{equation} \label{eq:equiv}
\mathrm{MEM}(K^*) \leftrightarrow \mathrm{VAL}(K), \qquad \mathrm{SEP}(K^*) \leftrightarrow \mathrm{VIOL}(K),
\end{equation}
where $\mathrm{MEM}(K^*) \leftrightarrow \mathrm{VAL}(K)$ means we can implement a weak validity oracle for $K$ using a single query to a weak membership oracle for $K^*$, and vice versa. Since $\mathrm{VIOL}(K)$ and $\mathrm{OPT}(K)$ are equivalent up to reductions that use $\widetilde \Theta(1)$ queries (via binary search), this justifies the central symmetry of Figure~\ref{fig:results}, because it shows that algorithms that implement $\mathrm{VIOL}(K)$ given $\mathrm{VAL}(K)$ are equivalent to algorithms that implement $\mathrm{SEP}(K^*)$ given $\mathrm{MEM}(K^*)$, and similarly algorithms that implement $\mathrm{SEP}(K)$ given $\mathrm{VIOL}(K)$ are equivalent to algorithms that implement $\mathrm{VIOL}(K^*)$ given $\mathrm{SEP}(K^*)$.

Gr{\"o}tschel, Lov\'asz, and Schrijver~\cite[Section~4.4]{grotschel1988GeomAlgAndConvOpt} showed that the weak membership problem for $K^*$ can be solved using a single query to a weak validity oracle for $K$, and that the weak separation problem for $K^*$ can be solved using a single query to a weak violation oracle for $K$. Using similar arguments one can show the reverse directions as well, which justifies~\eqref{eq:equiv}. Here we only motivate the equivalences between the above-mentioned weak oracles by showing the equivalence of the strong oracles (i.e., where $\rho$ and $\eps$ are~$0$). 

\paragraph{Strong membership on $K^*$ is equivalent to strong validity on $K$.}
First, for a given vector $c \in \R^n$ and a $\gamma >0$ observe the following:
\[
\frac{c}{\gamma} \not \in \mathrm{int}( K^*) \quad \Longleftrightarrow \quad  \exists y \in K \text{ s.t.\ } \langle c/\gamma,y \rangle \geq 1 \quad \Longleftrightarrow \quad \exists y \in K \text{ s.t. } \langle c,y \rangle \geq \gamma.
\]
Hence, a strong membership query to $K^*$ with a point $c$ can be implemented by querying a strong validity oracle for $K$ with the vector $c$ and the value $1$. Likewise, a strong validity query to $K$ with a point $c$ and value\footnote{Observe that validity queries with value $\gamma \leq 0$ can be answered trivially, since $0 \in K$.\label{foot:negQuery}} $\gamma>0$ can be implemented using a strong membership query to $K^*$ with $c/\gamma$.

\paragraph{Strong separation on $K^*$ is equivalent to strong violation on $K$.}
To implement a strong separation query on $K^*$ for a vector $y\in \R^n$ we do the following. Query the strong violation oracle for $K$ with $y$ and the value $1$. If the answer is that $\langle y,x \rangle \leq 1$ for all $x \in K$, then $y \in K^*$. If instead we are given a vector $x \in K$ with $\langle y,x \rangle \geq 1$,  
then $x$ separates $y$ from $K^*$ (indeed, for all $z \in K^*$, we have $\langle z,x \rangle \leq 1 \leq \langle y,x \rangle$). 

For the reverse direction, to implement a strong violation oracle for $K$ on the vector $c$ and value\textsuperscript{\ref{foot:negQuery}} $\gamma>0$ we do the following. Query the strong separation oracle for $K^*$ with the point $c/\gamma$. If the answer is that $c/\gamma \in K^*$ then $\langle c,x \rangle \leq \gamma$ for all $x \in K$. If instead we are given a non-zero vector $y \in \R^n$ that satisfies $\langle c/\gamma, y \rangle \geq \langle z, y \rangle$ for all $z \in K^*$, then $\tilde y = y/\langle c/\gamma, y \rangle$ will be a valid answer for the strong violation oracle for $K$. Indeed, we have $\tilde y \in K$ because  $\langle z, \tilde y \rangle \leq 1$ for all $z \in K^*$ and $K = (K^*)^*$, and by construction $\langle c, \tilde y \rangle = \gamma$.

\section{Discussion and future work}

We mention several open problems for future work:
\begin{itemize}
\item Our current implementation of an optimization query using $\bOt{n}$ quantum membership queries is quadratically better than the best known classical randomized algorithm, which uses roughly $n^2$ membership queries. However, to the best of our knowledge it is open whether this quadratic classical bound is optimal (a quadratic classical lower bound is known for \emph{deterministic} algorithms~\cite{yao1975ComputingTheMinimaOfQuadraticForms}).

	\item Can we improve our $\Omega(\sqrt{n})$ lower bound on the number of separation (or membership) queries needed to implement an optimization oracle when our algorithm knows a point in~$K$? We conjecture that the correct bound is $\tilde{\Theta}(n)$, in which case knowing a point in~$K$ does not confer much benefit for query complexity. 
	
	
	\item Are there interesting convex optimization problems where separation is much harder than membership for classical computers?\footnote{Moment polytopes are promising candidates for such examples. Recently an efficient weak membership oracle was constructed by Bürgisser et al.~\cite{burgisser2018EffTensorScalingMomentPolytopes} for a class of these polytopes. However, to the best of our knowledge it is unknown how to directly implement separation oracles for them, so one might get a quantum speed-up for implementing separation oracles using few queries to their membership oracle.}  Such problems would be good candidates for quantum speed-up in optimization in the real, non-oracle setting of time complexity. It is known that given a deterministic algorithm for function evaluation, an algorithm with roughly the same complexity can be constructed to compute the gradient of that function~\cite{griewank2008EvalDerivatives}. Hence for strong, deterministic oracles, separation is not much harder than membership queries. This, however, still leaves weak / randomized / quantum membership oracles to be considered.
	
	\item The algorithms that give an $\bOt{n}$ upper bound on the number of separation queries for optimization (for example~\cite[Theorem~42]{lee2015FasterCuttingPlaneConvexOpt}) give the best theoretical results for many convex optimization problems. However, due to the large constants in these algorithms they are rarely used in a practical setting. A natural question is whether the algorithms used in practice lend themselves to quantum speed-ups. Very recent work by Kerenidis and Prakash~\cite{kerenidis2018QIntPoint} on quantum interior point methods is a first step in this direction.
\end{itemize}

\paragraph{Acknowledgments.}
We thank Shouvanik Chakrabarti, Andrew Childs, Tongyang Li, and Xiaodi Wu for sending us a preliminary version of their paper~\cite{chakrabarti2018QuantumConvexOpt}, and for useful comments and coordination between our two papers.
AG thanks M\'{a}ri\'{o} Szegedy for insightful discussions. Many thanks to the anonymous referees of QIP'19 and Quantum for their constructive comments.

JvA and SG are supported by the Netherlands Organization for Scientific Research (NWO), grant number 617.001.351.
AG and RdW are supported by ERC Consolidator Grant 615307-QPROGRESS. 
RdW is also partially supported by NWO through Gravitation-grant Quantum Software Consortium - 024.003.037, and through QuantERA project Quant\-Algo 680-91-034.

\bibliographystyle{alphaUrlePrint}
\bibliography{Bibliography}

\appendix

\section{Quantum gradient computation using relational oracles} \label{sec:app}

In this appendix we extend the result of Corollary~\ref{cor:corollaryJordan} to functions given by a relational input oracle.
As a direct consequence this shows that the algorithm from Theorem~\ref{mainthrm:precise} also works when the input is given as a relational membership oracle instead of a standard oracle.

\begin{definition}[Unitary $\delta$-approximator]\label{def:unideltaapprox}
	Let $X$ be a finite set and let $Y$ denote a set of fixed-point $b$-bit numbers. Let $f\colon X\to Y$ be a function. 
	We say that a relational quantum oracle $U$ on $X$ is a \emph{$b$-bit unitary $\delta$-approximator of $f$} if the valid answers for each $x\in X$ differ at most $\delta$ from $f(x)$ (i.e., $\mathcal{F}(x)=\{y\in Y\colon |f(x)-y|\leq \delta\}$), and the success probability is at least $\frac{2}{3}$.
\end{definition}

\begin{corollary}[Gradient computation using a unitary $\delta$-approximator]\label{cor:generalCorollaryJordan}
	Let $\delta,B,r,c\in\R$, $\rho\in(0,1/3]$. Let $x_0,g\in\R^n$ with $\nrm{g}_\infty\leq \frac{B}{r}$. 
	Let $m:=\left\lceil\log_2\left(\frac{B}{28\pi\delta}\right)\right\rceil$ and suppose $f:\left(x_0+r G_m^n\right)\to \R$ is such that 
	\begin{equation*}
	\left|f(x_0+r x)-\ipc{g}{rx}-c\right|\leq \delta
	\end{equation*}
	for 99.9\% of the points $x\in G_m^n$, and
	we have access to $U$, an $\bigO{\log\left(\!\frac{B}{\delta}\!\right)}$-bit unitary $\delta$-approximator of $f$ over the domain $\left(x_0+r G_m^n\right)$.
	Then we can compute a vector $\tilde{g}\in\R^n$ such that 
	$$\Pr\left[\,\nrm{\tilde{g}-g}_\infty>\! \frac{8\cdot 42\pi \delta}{r}\right]\leq \rho,$$
	with $\bigO{\!\log\!\big(\frac{n}{\rho}\big)\!}$ queries to $U$ and $U^\dagger$ 
	and gate complexity $\bigO{\!n\log\!\big(\frac{n}{\rho}\big)\!\log\!\big(\!\frac{B}{\delta}\!\big)\!\log\!\log\!\big(\frac{n}{\rho}\big)\!\log\!\log\!\big(\!\frac{B}{\delta}\!\big)\!}\!$.
\end{corollary}

\begin{proof}
	The algorithm is the same as in the less general Corollary~\ref{cor:corollaryJordan} presented in Section~\ref{sec:quantumgrad}, we just  need to analyze it a bit more carefully. The main idea is still to implement an approximate version of the phase oracle $\mathrm{O}:\ket{x,0,0}\mapsto e^{2\pi i \frac{M}{3B} f(x_0+r x)}\ket{x,0,0}$, and then use Jordan's gradient computation algorithm. We approximate $\mathrm{O}$ by first approximately computing $f$ using $U$, then applying\footnote{If $y$ is a $b$-bit fixed-point binary number, then this can be implemented using $b$ single-qubit phase gates as follows: we can assume without loss of generality that $y=a_0 + a\cdot\sum_{j=1}^{b}y_j 2^j$ for some fixed $a_0, a\in\R$. Then $e^{2\pi i \frac{M}{3B} y}=e^{2\pi i \frac{M}{3B} a_0}\prod_{j=1}^{b}e^{2\pi i \frac{M}{3B}a y_j 2^j}$. The global phase is irrelevant, and the other phase factors can be implemented by using $b$ single-qubit phase gates, each acting as $\ket{y_j}\mapsto e^{2\pi i \frac{M}{3B}a y_j 2^j}\ket{y_j}$. } a controlled phase operation $\mathrm{cP}$ acting as $\mathrm{cP}\colon\ket{y}\mapsto e^{2\pi i \frac{M}{3B} y}\ket{y}$ (where $M=\frac{3B}{84\pi\delta}$ as in the proof of Corollary~\ref{cor:corollaryJordan}), and finally applying $U^\dagger$ to approximately uncompute $f$.
	
	We can assume without loss of generality that our unitary $\delta$-approximator is such that the probability of
	$\left|f(x)-y\right|> \delta$ is at most $\frac{1}{1200}$. If this is not the case, we can improve the success probability by querying $U$ a few times and taking the median of the results.
	
	Let us define $\mathcal{F}(x):=\{y\in Y \colon \left|f(x)-y\right|\leq \delta\}$ as in Definition~\ref{def:unideltaapprox}.
	Observe that 
	\begin{align*}
	\nrm{\mathrm{O}\ket{x,0,0} - U^\dagger \left(I\otimes \mathrm{cP} \otimes I\right) U\ket{x,0,0}}^2
	&=\nrm{\left(I\otimes(e^{2\pi i \frac{M}{3B} f(x_0+r x)}I-\mathrm{cP})\otimes I\right) U\ket{x,0,0}}^2\\
	&=\nrm{\sum_{y\in Y}\left(e^{2\pi i \frac{M}{3B} f(x_0+r x)}-e^{2\pi i \frac{M}{3B} y}\right)\alpha_{x,y}\ket{x,y,\psi_{x,y}}}^2.	
	\end{align*}
	We bound the above quantity in two parts using the triangle inequality as follows:
	\begin{align*}
	\!\nrm{\sum_{y\in Y\setminus \mathcal{F}(x) }\left(e^{2\pi i \frac{M}{3B} f(x_0+r x)}-e^{2\pi i \frac{M}{3B} y}\right)\alpha_{x,y}\ket{x,y,\psi_{x,y}}}^2
	&\!\!\!\leq \!\!\! \sum_{y\in Y\setminus \mathcal{F}(x) }|2\alpha_{x,y}|^2\leq \frac{1}{300};\\[2mm]
	\nrm{\sum_{y\in  \mathcal{F}(x) }\left(e^{2\pi i \frac{M}{3B} f(x_0+r x)}-e^{2\pi i \frac{M}{3B} y}\right)\alpha_{x,y}\ket{x,y,\psi_{x,y}}}^2
	&\!\!\!\leq \!\sum_{y\in \mathcal{F}(x) }\left|2\pi i \frac{M}{3B}(f(x_0+r x)-y)\alpha_{x,y}\right|^2\\
	&\!\!\!\leq \!\sum_{y\in Y_x }\left|2\pi i \frac{M}{3B}\delta\right|^2\left|\alpha_{x,y}\right|^2\\
	&\!\!\!\leq \left|2\pi i \frac{M}{3B}\delta\right|^2
	=\frac{1}{42^2}.
	\end{align*}
	Thus for all $x\in G_m^n$ we have that
	\begin{equation}\label{eq:closeProbability}
	\nrm{\mathrm{O}\ket{x,0,0} - U^\dagger \left(I\otimes \mathrm{cP} \otimes I\right) U\ket{x,0,0}}\leq \sqrt{\frac{1}{300}+\frac{1}{42^2}}< \frac{1}{16}.
	\end{equation}
	
	We can assume without loss of generality that our approximate phase oracle does not change the value of the input register. Otherwise we can just copy $\ket{x}$ to another register, then apply our approximate phase oracle on the second copy, then (approximately) erase the second copy of $\ket{x}$ using mod $2$ bitwise addition with the first copy. Under this assumption by \eqref{eq:closeProbability} we get that 
	\begin{equation}\label{eq:closeProbabilitySuper}
	\nrm{\mathrm{O}\ket{\psi} - U^\dagger \left(I\otimes \mathrm{cP} \otimes I\right) U\ket{\psi}}< \frac{1}{16}, \text{ for any quantum state }\ket{\psi}=\sum_{x\in G_m^n}\alpha_x\ket{x,0,0}.
	\end{equation}
	
	From now on the proof is the same as the proof of Corollary~\ref{cor:corollaryJordan}. In that proof we showed that if we use the phase oracle $\mathrm{O}$ in Jordan's gradient computation algorithm, then we would get a gradient estimate where each individual coordinate has the required approximation quality with probability at least $\frac{2}{3}$. Equation~\eqref{eq:closeProbabilitySuper} implies that if instead we use our approximate implementation of the phase oracle, $U^\dagger \left(I\otimes cP \otimes I\right) U$, then the outcome probability distribution changes by at most $\frac{1}{16}$ in total variation distance. So one run of Jordan's algorithm using this approximate phase oracle still outputs a vector $v\in\mathbb{R}^n$ such that
	$$
	\Pr\left[\left|g_i-\frac{3B}{r}v_i\right|> \frac{8\cdot 42 \pi \delta}{r}\right]\leq \frac{1}{3}+\frac{1}{16} < \frac{2}{5} \text{ for every }i\in[n].
	$$
	
	As in the proof of Corollary~\ref{cor:corollaryJordan}, repeating the whole procedure $\bigO{\log(\frac{n}{\rho})}$ times, and taking the median of the resulting vectors coordinatewise, gives a gradient approximator $\tilde{g}$ of the desired quality. The gate complexity analysis follows from \cite[Theorem~21]{gilyen2017OptQOptAlgGrad}, noting that each controlled phase operation $cP$ can be implemented using $\bigO{\log\left(\!\frac{B}{\delta}\!\right)}$ single-qubit phase gates.
\end{proof}

\end{document}